\documentclass[11pt]{article}
\usepackage{amsmath}
\usepackage{amsthm}
\usepackage{amssymb}
\usepackage{amsfonts}
\usepackage{mathtools}
\usepackage{bbm}
\usepackage{bm}
\usepackage{booktabs}
\usepackage{graphicx}
\usepackage{subfigure}
\usepackage{multirow}
\usepackage{longtable}
\usepackage{adjustbox}
\usepackage{longtable}
\usepackage{algorithm2e}

\usepackage{microtype}
\usepackage{xspace}
\usepackage{verbatim}
\usepackage{color}
\usepackage{comment}

\usepackage{natbib}
\usepackage{hyperref}
\usepackage[capitalize,noabbrev]{cleveref}
\usepackage{url}

\usepackage{authblk}

\makeatletter
\g@addto@macro{\UrlBreaks}{\UrlOrds}
\makeatother
\usepackage{threeparttable}

\newtheorem{definition}{Definition}
\newtheorem{remark}{Remark}
\newtheorem{theorem}{Theorem}

\usepackage[textsize=tiny]{todonotes}
\usepackage{pdflscape}
\def\R{\mathbb{R}}

\newcommand{\norm}[1]{\lVert{#1}\rVert}

\newcommand{\PP}[1]{\mathbb{P}\left\{{#1}\right\}} 
\newcommand{\EE}[1]{\mathbb{E}\left[{#1}\right]} 

\def\R{\mathbb{R}}
\def\Z{\mathbf{Z}}
\def\W{\mathbf{W}}

\newcommand\eqd{\stackrel{\mathclap{\normalfont\mbox{d}}}{=}}

\newcommand{\X}{\mathbf{X}}
\newcommand{\Y}{\mathbf{Y}}
\newcommand{\V}{\mathbf{V}}
\newcommand{\M}{\mathbf{M}}

\newcommand{\bmb}{\mathbf{b}}
\newcommand{\bms}{\mathbf{s}}
\newcommand{\bmbeta}{\bm{\beta}}
\newcommand{\bmmu}{\bm{\mu}}
\newcommand{\bmSigma}{\bm{\Sigma}}

\DeclareMathOperator{\diag}{diag}

\newcommand{\Xp}{\widetilde{X}}

\newcommand{\indep}{\perp \!\!\! \perp}

\def\bSig\mathbf{\Sigma}

\renewcommand{\hat}{\widehat}
\renewcommand{\tilde}{\widetilde}

\title{Model-free High Dimensional Mediator Selection with False Discovery Rate Control}

\author[1]{Runqiu Wang}
\author[1]{Ran Dai\thanks{ran.dai@unmc.edu}}
\author[2]{Jieqiong Wang}
\author[1]{Kah Meng Soh}
\author[2]{Ziyang Xu} 
\author[2]{Mohamed Azzam} 
\author[1]{Hongying Dai} 
\author[1]{Cheng Zheng\thanks{cheng.zheng@unmc.edu}}

\affil[1]{Department of Biostatistics, University of Nebraska Medical Center, Omaha, Nebraska, U.S.A.}

\affil[2]{Department of Neurological Science, University of Nebraska Medical Center, Omaha, Nebraska, U.S.A.}

\begin{document}

\maketitle



\begin{abstract}
   There is a challenge in selecting high-dimensional mediators when the mediators have complex correlation structures and interactions. In this work, we frame the high-dimensional mediator selection problem into a series of hypothesis tests with composite nulls, and develop a method to control the false discovery rate (FDR) which has mild assumptions on the mediation model. We show the theoretical guarantee that the proposed method and algorithm achieve FDR control. We present extensive simulation results to demonstrate the power and finite sample performance compared with existing methods. Lastly, we demonstrate the method for analyzing the Alzheimer's Disease Neuroimaging Initiative (ADNI) data, in which the proposed method selects the volume of the hippocampus and amygdala, as well as some other important MRI-derived measures as mediators for the relationship between gender and dementia progression.

%
%

\textbf{Keywords}:FDR; High-dimensional mediators; imaging data; knockoff 
\end{abstract}
\section{Introduction}
In social and biomedical research, mediation analysis is frequently employed to investigate how exposure influences an outcome through intermediary variables. The rise of studies that collect high-dimensional data, such as neuroimaging \citep{poldrack2024handbook,dan2024image,penny2011statistical} or high-throughput omics data \citep{subramanian2005gene, love2014moderated}, has sparked interest among scientists in uncovering mediators within these datasets to address complex scientific questions involving various exposures and outcomes. 

{\color{black}As a motivating example, we investigate the impact of gender on dementia progression using data from the Alzheimer’s Disease Neuroimaging Initiative (ADNI). Evidence suggests that men and women exhibit differences in the development and progression of Alzheimer’s disease (AD). However, the factors underlying these gender differences are still poorly understood \citep{pike2017sex}. Understanding these neurobiological mechanisms is critically important for developing targeted interventions and precision medicine approaches, as female and male populaitons may require distinct therapeutic strategies \citep{cahill2006sex,mccarthy2011reframing,cui2023enhancing}. The current knowledge gap limits our ability to accurately predict disease trajectories and optimize treatment outcomes across diverse patient populations. Moreover, biological systems rarely operate through simple linear mechanisms. Instead, complex nonlinear relationships and interactions between mediators are commonly observed in neurological processes. This complexity is particularly relevant in the relationship between gender and AD, where multiple neural pathways can work together or against each other to affect disease progression through higher-order interactions \citep{jack2013tracking,insel2016assessing}. To address these challenges, we proposed to identify mediators from high-dimensional structural MRI-derived imaging measures to better illustrate the mechanisms driving AD progression differences between the gender.} 

{\color{black}From a statistical perspective, mediation analysis has traditionally been developed and extensively applied to low-dimensional settings over the past decade \citep{baron_moderatormediator_1986, vanderweele_explanation_2015}. In such scenarios, mediation analysis aims to assess whether the relationship between an exposure and an outcome is mediated by a low-dimensional set of mediators. More recently, advancements in data collection techniques across fields such as imaging, genetics, and omics research motivate the development of methods tailored for high-dimensional mediation analysis \citep{sampson2018fwer, zhang_estimating_2016,perera2022hima2,bai2024efficient,dai2024_mediation}. These recent methods mostly focus on the linear model cases without considering interactions among the mediators. The mediation effect for the $j$th mediator is quantified as $\alpha_j \cdot \beta_j$, where $\alpha_j$ represents the marginal association between the $j$th mediator and the exposure (the path-a effect), and $\beta_j$ denotes the conditional association between the outcome and the $j$th mediator given the exposure (the path-b effect). However, this mediation model set-up does not account for the interactions and nonlinear relationships between the outcome and mediators, which are commonly expected in the neuroimaging data.}
Another challenge in high-dimensional mediation analysis is how to ensure reproducibility in the identification of mediators. In practical high-dimensional exploratory research applications, where the signals are often rare and subtle, controlling the false discovery rate (FDR) is generally preferred over the family-wise error rate (FWER) to enhance selection power while managing type-I errors effectively \citep{benjamini_controlling_1995}. To control FDR in high-dimensional mediation studies, \citet{zhang_estimating_2016} performs a joint significance test with the subset of mediators after a screening and penalization step (HIMA). {\color{black}\cite{sampson2018fwer} proposed the Multiple Comparison Procedure (MCP), using robust statistical models to test direct and indirect associations to control FDR asymptotically for the independent mediators.} 
\citet{perera2022hima2} further proposed a hybrid method using debiased lasso post-selection inference and the JS-mixture test, which assumes a linear model in path-b and also relies on independence or weak dependence structure of the mediators {\color{black}(HIMA2). More recently, \cite{bai2024efficient} proposed an efficient three-step process designed to reduce computational overhead compared to HIMA2 while incorporating a robust FDR control mechanism (eHIMA). }\cite{dai2024_mediation} proposed an iterative approach by controlling FWER in path-a and FDR in path-b, which allows for complex mediator correlations and has theoretical FDR control guarantees {\color{black}(mediationFDR)}. {\color{black}This method has a stringent FWER control applied for path-a, which becomes increasingly conservative as the dimensionality of mediators grows.

{\color{black}With the ADNI structural MRI data as candidate mediators, the existing methods present limitations. First, path-a effects are typically small due to subtle gender differences in brain structure and function compared to the large within-gender variability. Second, interactions between various structural MRI measures potentially influence cognitive outcomes. 
For example, a study found that in individuals with lower gray matter (GM) volumes, increased functional connectivity between the medial prefrontal cortex and the lateral intraparietal cortex was associated with better memory performance. However, this association was not found in individuals with higher GM volumes. This means that GM volume might change how functional connectivity and memory are related \citep{he2012influence}.
Furthermore, another study found that the baseline white matter hyperintensity volume (WMHV) worked with age to affect CDR-SB scores in females but not in males, showing that brain structure changes might affect memory loss differently in men and women \citep{schweitzer2024sex}. Therefore, a more advanced nonlinear mediation framework is required to capture the true mechanistic pathways in AD progression.}


To address our data challenges, we develop a novel high-dimensional mediator selection framework with theoretical FDR control guarantees {\color{black}asymptotically}. We first frame the mediator selection procedure into a multiple testing problem with composite null hypotheses. Then we use the simultaneous knockoff \citep{dai2023_simultaneous} idea to construct the filter statistics for the mediators. The major challenge of this extension is how to handle the unavoidable dependence between path-a and path-b from the exposure variable and potential confounding variables. Our high-level strategy is to decompose path-b to obtain an independent component from path-a.
The main contributions of this paper are listed below:\\
1. We formulate a general multiple hypothesis testing framework for mediator selection from a high-dimensional candidate set with potentially nonlinear relationships with the exposure or the outcome. \\
2. For the traditional mediation model, when both path-a and path-b are with linear relationships, we develop a simplified knockoff generating using partial least squares to decompose the path-b to obtain an independent component of path-a. \\
3. We show the theoretical FDR control guarantees with mild model assumptions for path-b.\\
4. We perform extensive simulations to demonstrate the FDR control performance and power of our proposed method, and compare with existing methods. \\
5. We applied our proposed method to analyze the ADNI data. 

\section{Method}\label{sec.method}
\subsection{Model settings}\label{sec:setting}
 We develop our method under a general multiple testing problem for mediator selection \citep{dai2024_mediation}. Let $Y \in \R$ be a one-dimensional response, $X \in \R$ be a one-dimensional exposure 
 , and $\mathbf{M} = (M_1,\cdots, M_p) \in \R^{p}$ is the high dimensional mediators, $\mathbf{V} \in \R^{d}$ are low dimensional confounding variables. For $j\in \{1,\cdots, p\} := [p]$, we define the null hypotheses for path-a: $H_{aj}$ and path-b: $H_{bj}$:
\begin{align} \label{eqn:linear_mediation_model} 
    &\text{path-a:} ~~ H_{aj}: M_j \indep X | \mathbf{V}, \nonumber \\ &\text{path-b:}~~ H_{bj}: Y \indep M_j | M_{-j},X, \mathbf{V}, \nonumber \text{where $M_{-j} := \{M_k: k \in [p], k\neq j\}$}.
\end{align}

Our goal is to test the composite null hypotheses
$H_{0j} = H_{aj}\cup H_{bj} ~\text{for}~ j \in [p]$. We define $M_j$ as a mediator if we reject the hypothesis $H_{0j}$. 

As an example, traditional mediation analysis model assumes linear relationships for both path-a and path-b, i.e. \citet{mackinnon_introduction_2008}.
\begin{equation}\label{eqn:M_linear}
\text{path-a:}\quad M_j = c_{1j} + X\alpha_j + \mathbf{V}^\top \boldsymbol{\eta}_{1j} + \epsilon_{1j}, \quad \textnormal{for}~ j \in [p],
\end{equation}
where $c_{1j} \in \R$ is the intercept,  $\alpha_j \in \R$, $\boldsymbol{\eta}_{1j} \in \R^{d}$ are the coefficients for the exposure $X$ and confounding variable $\mathbf{V}$ respectively and $\epsilon_{1}=(\epsilon_{11},\cdots,\epsilon_{1p}) \in \R^p$ are the multivariate Gaussian noises.
\begin{equation}\label{eqn:Y_linear}
\text{path-b:}\quad Y = c_2 + X\gamma + \mathbf{M}^\top\boldsymbol{\beta} + \mathbf{V}^\top\boldsymbol{\eta}_2 + \epsilon_{2}, \end{equation} where $\mathbf{V} \in \R^{d}$ are low dimensional confounding variables. $c_2 \in \R$ is the intercept, $\gamma \in \R$, $\boldsymbol{\beta}\in \R^p$, $\boldsymbol{\eta}_2 \in \R^{d}$ are the coefficients. 


 Our framework does not restrict our model to be linear in either path-a or path-b. For path-a, the marginal models for each $M_j$ does not need to be linear. For example, we can consider replacing equation \eqref{eqn:M_linear} to a generative additive model $g_j^{-1}(\mathbb{E}[M_j])=h_1(X)+h_2(\V)$ or a more general non-parametric form $M_j=g_j(X,\V,\epsilon_{1j})$ {\color{black}with arbitrary pre-specified function $g_j$s}.
 
 For path-b, equation (\ref{eqn:Y_linear}) can be relaxed to allow for interactions between the mediators:

\begin{equation}\label{eqn:Y_interacton}
    Y = c + X\gamma + \mathbf{M}^\top\boldsymbol{\beta} + \bm{\Delta}^\top\boldsymbol{\delta}+\mathbf{V}^\top\boldsymbol{\eta}_{2} + \epsilon_{2}
\end{equation}
where $\bm{\Delta}=(M_1M_2,\cdots,M_{p-1}M_p) \in \R^{p(p-1)/2}$ is for two-way interaction terms and $\boldsymbol{\delta}$ is its coefficient.

Our method also allows for general nonlinear model forms of path-b to accommodate various types of response variables, including categorical and count data. For example, we can have generalized additive models like
$g^{-1}(\mathbb{E}[Y])=f_1(X)+\sum_{j=1}^pf_{2j}(M_j)+f_3(\V)+\epsilon_2$
or a nonparametric model $Y=f(X,\M,\Z,\V,\epsilon_2)$ {\color{black}where $f_1$, $f_{2j}$, $f_3$, $f$ are unknown functions}. 

 Under our general model setting, we frame the mediator selection problem as a multiple testing problem, where for each $j \in \{1,\cdots,p\}$, $H_{0j} =  H_{aj} \cup H_{bj}$ is a union null hypothesis, and $\mathcal{H}  = \{k\in [p]: H_{0k} ~\text{is true}\}$. Our goal is to identify a selection set $\widehat{S} \subset [p]$ and control for the FDR defined as

\begin{equation} \label{eqn:fdr_def}
    \text{FDR}(\hat{S}) = \EE{\frac{|\hat{S} \cap \mathcal{H}|}{|\hat{S}|\vee 1}}.
\end{equation}

\subsection{Preliminaries}
{\color{black}Before illustrating our algorithm, we briefly review the knockoff method, which is a powerful statistical technique to control the FDR in variable selection problems. The fundamental idea behind knockoff methods \citep{barber2015,candes2018} is to generate artificial "knockoff" variables that retain their inner correlation structure but are independent of the response variable. By comparing the behavior of original variables against their knockoffs, we can estimate and control the false discovery rate. \citet{barber2020} further demonstrated the robustness of the Model-X knockoff to errors in the underlying assumptions regarding the distribution of $X$. Additionally, \citet{huang2020} relaxed the assumptions of the Model-X knockoff, showing that FDR control is achievable as long as the distribution of $X$ is known up to a parametric model. \citet{romano2019} devised a Deep Knockoff machine employing deep generative models, while \citet{liu2019} developed a Model-X generating method utilizing deep latent variable models. More recently, \citet{bates2020} proposed an efficient general metropolized knockoff sampler, and \citet{spector2020} suggested constructing knockoffs that minimize constructability.

For the simultaneous variable selection from multiple experiments, the knockoff framework has been extended to examine multiple union null hypotheses across independent studies to address the challenges in reproducibility research for heterogeneous populations \citep {dai2023_simultaneous,li2022searching,wang2023controlling}. The idea of these approaches is to construct the knockoff copy and test statistics for each experiment individually, and then combine the test statistics through specially designed functions (i.e., the One Swap Flip Sign Function); therefore, exact FDR control can be provided when testing the union null hypotheses. In the causal mediator analysis, a variable can only be a mediator if it is both associated with the treatment (in path-a) and the outcome (in path-b). By conceptualizing mediation as a union null hypothesis testing problem across two conditional independence relationships, the simultaneous knockoffs framework provides an idea to control the false discovery rate when selecting high-dimensional mediators.} 

{\color{black} We also introduce some preliminary definitions for the Knockoff methods.

\begin{definition}[Swapping]\label{def:SWAP}
For a set $S \subseteq [p]$, and for a vector $\mathbf{U} = (U_1, \ldots, U_{2p}) \in \mathbb{R}^{2p}$, $\mathbf{U}_{\text{Swap}(S)}$ indicates the swapping of $U_j$ with $U_{j+p}$ for all $j \in S$.
\end{definition}

\begin{definition}[One swap flip sign function (OSFF)] \label{def:OSFF}
A function $f : \mathbb{R}^{2pK} \rightarrow \mathbb{R}^{p}$ is called an OSFF if it satisfies that for all $k \in [K]$ and all $S \subseteq [p]$,
\begin{align*}
f(&[\mathbf{Z}^1, \tilde{\mathbf{Z}}^1], \ldots, 
[\mathbf{Z}^k, \tilde{\mathbf{Z}}^k]_{\text{Swap}(S)}, \ldots, 
[\mathbf{Z}^K, \tilde{\mathbf{Z}}^K]) 
= f([\mathbf{Z}^1, \tilde{\mathbf{Z}}^1], \ldots, 
[\mathbf{Z}^k, \tilde{\mathbf{Z}}^k], \ldots, 
[\mathbf{Z}^K, \tilde{\mathbf{Z}}^K]) \odot \epsilon(S),
\end{align*}
where $\mathbf{Z}^k, \tilde{\mathbf{Z}}^k \in \mathbb{R}^{p}$ for $k \in [K]$, $\epsilon(S) \in \mathbb{R}^{p}$, and $\epsilon(S)_j = -1$ for all $j \in S$ and $\epsilon(S)_j = 1$ otherwise. Here, $\odot$ denotes the Hadamard product.
\end{definition}
}

\subsection{Mediator Selection algorithm}\label{sec:algorithm}

We design the following procedure to control for the FDR of the selected set of mediators {\color{black} as defined in equation \eqref{eqn:fdr_def}}.

\begin{itemize}
    \item Step 1: \textbf{Construct knockoff-based test statistics for path-a.}\quad 
    \paragraph{1.1 Path-a knockoff generation} For path-a, for each $j \in [p]$, we first generate knockoff $\tilde{M}_j$ for $M_j$, where $\tilde{M}_j$ is independent of $X$ and has identical distribution with $M_j$ conditioning on $\V$. We can generate knockoff $(\tilde{M}_j, \tilde{\V})$ for $(M_j, \V)$ using the Model-X knockoff construction approach \citep{candes2018}. Depending on the specific data type of $M_j$, we can use the Gaussian approach, second-order approach, sequential approach, or even deep learning-based approaches \citep{kormaksson2021sequential, romano2019} for the knockoff construction. {\color{black}More details for the Model-X Knockoff assumptions and constructions methods can be found in Web Appendix A.1.}
    \paragraph{1.2 Path-a test statistics computation} We fit a regression model of $X$ on $(M_j,\tilde{M}_j,\V)$ and use the absolute value of the coefficients for $M_j$ and $\tilde{M}_j$, or some variable importance measures {\color{black} (i.e., {\color{black} variable} importance score or SHapley Additive exPlanations (SHAP) value)} for $M_j$ and $\tilde{M}_j$ as test statistics $Z^a_{j}$, $\tilde{Z}^a_{j}$. Denote $\Z^a = (Z^a_1, Z^a_2, \cdots, Z^a_p)$ and $\tilde{\Z}^a = (\tilde{Z}^a_1, \tilde{Z}^a_2, \cdots, \tilde{Z}^a_p)$, under different model assumptions, the path-a statistics can be shown to satisfy a key property that for any $S\in \mathcal{H}_a$, $(\Z^a, \tilde{\Z}^a) \eqd (\Z^a, \tilde{\Z}^a)_{Swap(S)}$, or a weaker version $(Z^a_j, \tilde{Z}_j^a) \eqd(\tilde{Z}^a_j, Z^a_j)$ for any $j\in \mathcal{H}_a$ \citep{candes2018}.

{\color{red}
    }
    \item Step 2: \textbf{Construction of knockoff based test-statistics for path-b.}\quad 
    For path-b, we will construct knockoff $\tilde{\M}^b$ for $\M$ conditioning on $\V$ and $X$. This can be done through generating fixed or Model-X knockoff of $(\M,\V,X)$ as $(\tilde{\M}^b,\tilde{\V}^b,\tilde{X}^b)$. Denote $\tilde{\M}^b = (\tilde{M}_1, \tilde{M}_2, \cdots, \tilde{M}_p)$.
    We then fit penalized regression models or non-parametric regression models (e.g. Random Forests or a Deep Learning model) for $Y$ on $(\M,\tilde{\M}^b,\V,X)$ and use the absolute value of coefficients or certain importance measure of variables (e.g. importance score or SHAP value) for the variables $M_j$, $\tilde{M}_j$ as $Z^b_j$ and $\tilde{Z}^b_j$. Denote $(\Z^b, \tilde{\Z}^b) = (Z^b_1,\cdots,Z^b_p, \tilde{Z}^b_1, \cdots, \tilde{Z}^b_p)$, which can be shown to satisfy the requirement that for any $S\in \mathcal{H}_b$, $(\Z^b, \tilde{\Z}^b) \eqd (\Z^b, \tilde{\Z}^b)_{Swap(S)}$ \citep{candes2018}. {\color{black}Further details on the test statistics compatible with each knockoff construction method are provided in Web Appendix A.2.}

\begin{remark} \label{rmk:PLS}
    When the linear model is assumed for both path-a and path-b, we propose an alternative way to obtain knockoff-based statistics using partial least squares (PLS). Specifically, we first fit linear regression models for $Y$, $M_j$ on $X$ and $\V$ and denote the residuals as $r_Y$, $r^b_{Mj}$ respectively. Then we can generate the fixed or Model-X knockoff $\tilde{r}^b_{Mj}$ and run the regression of $r_X$ on $r^b_{Mj}$ and $\tilde{r}^b_{Mj}$ to obtain $(\Z^b, \tilde{\Z}^b)$.
\end{remark}

    \item Step 3: \textbf{Combine test statistics from path-a and path-b.}\quad After we construct statistics $(\Z^a, \tilde{\Z}^a)$ and $(\Z^b, \tilde{\Z}^b)$, using an OSFF as defined in Definition \ref{def:OSFF} to obtain test statistics $\W$. 
The simplest OSFF is $\W=(\Z^a-\tilde{\Z}^a)\odot (\Z^b-\tilde{\Z}^b)$, where $\odot$ is defined as the elementwise product. When a nonlinear model is assumed, there is no guarantee that $(\Z^a,\tilde{\Z}^a)\indep (\Z^b,\tilde{\Z}^b)$ but a data splitting will be able to solve the problem theoretically. In practice, such data splitting may not be necessary. Numerically based on our simulations, even without data splitting, the FDR are in general under control.
    
    
     \item Step 4: \textbf{Cut-off point set-up.}\quad 
     The property of $\W$ constructed from step 3 allows us to estimate the proportion of false discovery in our selection $\widehat{S}(t)$ as 
$\text{FDP}(t) \approx \widehat{\text{FDP}}(t) = \frac{\# W_j\leq -t}{(\# W_j \geq t) \vee 1}$, which supports the usage of the following stopping rules from the SeqSelect procedure \citep{barber2015} to determine the cut-off points for finite sample FDR control. The cut-off point 
\begin{equation*}
    \tau = \min\left\{t>\mathcal{W}_+: \frac{\#\{j:W_j \leq -t\}}{\#\{j:W_j\geq t\}\vee 1}\leq q \right\} ~~\textnormal{(Knockoff)},
\end{equation*}
controls the modified FDR (mFDR) (defined as $\text{mFDR} = \EE{\frac{|\{j\in \hat{S}\cap \mathcal{H}\}|}{|\hat{S}|+1/q_2}}$). A more conservative cut-off point 
\begin{equation*}
    \tau+ = \min\left\{t>\mathcal{W}_+: \frac{1+\#\{j:W_j \leq -t\}}{\#\{j:W_j\geq t\}\vee 1}\leq q\right\} ~~\textnormal{(Knockoff+)},
\end{equation*}
controls the FDR.
 \end{itemize}

With the above procedures, the final FDR in equation \eqref{eqn:fdr_def} or a modified FDR 
\begin{equation}\label{eqn:mfdr}
\textnormal{mFDR}=\EE{\frac{|\widehat{S}\cap \mathcal{H}|}{|\widehat{S}|+1/q}} \end{equation}
for mediators can be controlled at level $q$ when $p\rightarrow \infty$ and weakly dependence assumption.

\subsection{Quantification of the mediation effects}\label{sec:effect}
After selecting the set of mediators, we further make inference on the natural direct effect (NDE) and natural indirect effect (NIE) under our general model settings. 

{\color{black} Under the traditional mediation models with the assumptions of linearity, no unmeasured confounding, and no exposure-mediator interaction \eqref{eqn:M_linear}-\eqref{eqn:Y_linear}, the NDE can be directly defined as $\gamma$, and the NIE=$\sum_{j=1}^p \alpha_j\beta_j$.}

{\color{black}For the general nonparametric model cases, under the counterfactual framework, denote  $M_k(x)$ for $k \in [p]$ as the $k$th potential mediator if the exposure value is set as $x$; and $\mathbf{M}(x) \in \R^{p}$ as the potential mediators if the exposure value is set as $x$, i.e. $\M(x) = (M_1(x), \cdots, M_p(x))$. For the outcome $Y$, define $Y(x,M_1(x_1'), \cdots, M_p(x_p'))$ as the potential outcome value when the exposure is set at $x$ and mediators are set as $(M_1(x_1'), \cdots, M_p(x_p'))$.  

We have $NDE=\mathbb{E}Y(1, \mathbf{M}(0)) - Y(0, \mathbf{M}(0))$, $NIE=\mathbb{E}Y(1, \mathbf{M}(1)) - Y(1, \mathbf{M}(0))$.

Following \citet{imai2010general}, we implement a Monte Carlo integration algorithm to  evaluate the integration of $
\widehat{\mathbb{E}}\left[Y\left(x, M_1(x_1), \dots, M_p(x_p)\right)\right]
$ numerically.\\
\textbf{Step 1:} Sample $(M_1(x_1), \ldots, M_p(x_p))$ for different values of $(x_1, \ldots, x_p) \in \{0,1\}^p$ for all $\boldsymbol{V}_i$;\\   
\textbf{Step 2:} Sample $Y(x, M_1(x_1), \ldots, M_p(x_p))$ for different values of $x \in \{0,1\}$ and $(x_1, \ldots, x_p) \in \{0,1\}^p$ for all $\boldsymbol{V}_i$;\\
\textbf{Step 3:} Compute $\widehat{\mathbb{E}} \left[ Y(x, M_1(x_1), \ldots, M_p(x_p)) \right]$ by averaging sampled\\ $Y(x, M_1(x_1), \ldots, M_p(x_p))$ and average over individual $i$. Then compute $\widehat{\text{NDE}}, \widehat{\text{NIE}},$ and $\widehat{\text{NIE}}_k$:

    \begin{align*}
  & \widehat{\text{NDE}} = \widehat{\mathbb{E}}\left[Y(1, \mathbf{M}(0)) - Y(0, \mathbf{M}(0))\right],
\widehat{\text{NIE}} = \widehat{\mathbb{E}}\left[Y(1, \mathbf{M}(1)) - Y(1, \mathbf{M}(0))\right],\text{ and}\\
   & \widehat{\text{NIE}}_k =  \widehat{\mathbb{E}}\left[ Y\left(0,\M_{1:k-1}(0) ,M_k(1), \M_{k+1:p}(0)  \right) \right] - \widehat{\mathbb{E}}\left[ Y\left(0, \M(0) \right) \right],
\end{align*}
\
where $\M_{1:k-1}(0) = (M_1(0),\cdots, M_{k-1}(0))$, and $\M_{k+1:p}(0) = (M_{k+1}(0),\cdots,M_p(0))$.

}

\subsection{Theoretical results}
\begin{theorem}\label{thm:1}
With the general model setting and assumptions stated as in Section \ref{sec:setting}, using the general algorithm in Section \ref{sec:algorithm}, with Model-X knockoff construction and data splitting, with the Knockoff filter, the final mFDR for mediators can be controlled at level $q$ and with the Knockoff+ filter, the final FDR for mediators can be controlled at level $q$ when $p\rightarrow \infty$ with a choice of path-a statistics that satisfy the weak dependence assumption. 
\end{theorem}
The proof of Theorem \ref{thm:1} is postponed to Web Appendix B. In Theorem \ref{thm:1} 
we provide the theoretical guarantee for the algorithm in Section \ref{sec:algorithm} under the assumption of the large number of weakly dependent path-a statistics. 
{\color{black}
\section{Simulation} \label{sec:simulation}
We conducted simulations to evaluate the performance of our proposed method under both linear and nonlinear model settings. We compare our proposed method with five other methods: MCP\citep{sampson2018fwer}, HIMA\citep{zhang_estimating_2016}, HIMA2\citep{perera2022hima2}, eHIMA\citep{bai2024efficient} and mediationFDR\citep{dai2024_mediation}. 

\subsection{Data Generation and Simulation Settings} 
\paragraph{Linear Model Settings} \label{sec:linear}
For each subject $i \in [n]$, the exposure variable $X_i \in \mathbb{R}$ and the covariate $V_{i} \in \mathbb{R}$ are generated from a multivariate normal distribution, with mean $\begin{pmatrix} 0 \\ 0 \end{pmatrix}$ and variance-covariance matrix $\begin{pmatrix} 1 & 0.5 \\ 0.5 & 1 \end{pmatrix}$. The correlated potential mediators $\M_i = (M_{i1}, \dots, M_{ip})^\top$ is generated using $\M_i = X_i \bm{\alpha} + V_{i}\bm{\eta}_{1}+\bm{\epsilon}_{1i}$, where $\bm{\alpha} = (\alpha_1, \dots, \alpha_p)^\top$,$\bm{\eta_1} = (\eta_{11}, \dots, \eta_{1p})^\top$ and $\bm{\epsilon}_{1i} \sim \mathcal{N}(\mathbf{0}, \bm{\Sigma})$, where $\bm{\Sigma}$ is a variance-covariance matrix with compound symmetry structure (i.e.,diagonal entries are 1 and off-diagonal entries are $\rho$).
The outcome $Y_i$ is generated as $Y_i = X_i\gamma + \mathbf{M}^\top\boldsymbol{\beta} +V_{i}\eta_{2}+\epsilon_{2i}$, where $\bm{\epsilon}_{2i} \sim \mathcal{N}(0, \sigma^2)$. 


\paragraph{Nonlinear Model Settings} \label{sec:nonlinear} We explore two nonlinear scenarios: (1) Two-way interactions exist in path-b; (2) The path-b model is an additive model with element-wise cosine transformed $M_j$'s.

\paragraph{\textbf{Setting 1}: Two-way interactions} {\color{black} We consider an scenario where path-a remain the same as the linear model setting; while in path-b, only interaction effects (i.e, interaction between mediators) are present.} The generation of the exposure variable $X_i$ and potential mediators $\M_i$ are the same as the linear setting. The outcome $Y_i$ is generated as $Y_i = X_i\gamma + Z_i \boldsymbol{\delta}+\epsilon_{2i}$, where $Z_i=(M_{i1}M_{i2},\dots,M_{i(p-1)}M_{ip})$, $\epsilon_{2i} \overset{i.i.d}{\sim} \mathcal{N}(0, \sigma^2)$.

\paragraph{\textbf{Setting 2}: Cosine transformation}
For this setting, the generation of the exposure variable $X_i$ and potential mediators $\M_i$ are still the same as the linear model setting. We extend our path-b model to a nonlinear model by applying the cosine function to each element of $\M$. Specifically, in path-b, the outcome $Y$ is generated as $Y_i = X_i\gamma + \sum_{j=1}^p \beta_j\cdot cos(M_{ij}) +\epsilon_{2i}$, where $\epsilon_{2i} \overset{i.i.d}{\sim} \mathcal{N}(0, \sigma^2)$. 

Details of the data generation process and parameter settings are provided in Appendix C.1, and method specifications in Appendix C.2.

\subsection{Simulatiom Results}

Figure\ref{fig:chg_linear} presents the simulation results for the empirical power and FDR in the linear setting. Our proposed method, Generalized Knockoff Mediator Selection (GKMS), demonstrates very good performance, achieving relatively high power and effective FDR control across all the experimented settings. The power reduces for all methods when the path-a signal is weak, but GKMS performs the best under such scenarios. HIMA2 also exhibits strong performance in terms of power. However, under several conditions, it fails to control the FDR (e.g. the setting with small effect sizes for path-a $a=0.2$ and path-b $b=0.2$, lower correlation structures $\rho=0.2$, and smaller dimensions $p=100$, Figure \ref{fig:chg_linear}). Both MCP and eHIMA achieve high power, but MCP can not control the FDR when mediators are correlated ($\rho>0$), and eHIMA fails to control FDR for higher dimensional mediators. HIMA and MediationFDR effectively control FDR in all scenarios but at the cost of reduced power. MediationFDR has low power when path-a signals are weak ($a < 0.3$) which is as expected as it controls the FWER in path-a.

Tables \ref{tab:sim2} and \ref{tab:sim3} summarize the empirical results for nonlinear models with interactions (\textit{\textbf{Setting 1}}) and models with cosine transformed mediators (\textit{\textbf{Setting 2}}). Across both settings, most methods control the false discovery rate (FDR) effectively, except for eHIMA. In terms of power, our proposed method demonstrates satisfactory performance, consistently achieving the good power across all settings, highlighting its signal detection capabilities in nonlinear settings. In contrast, all other methods show very low powers {\color{black} since our settings have approximately no main effects while focusing on interaction and nonlinearly transformed effects.} Although our proposed method (GKMS) shows a slight decrease in power as the correlation increases, it remains stable in FDR control. Additionally, in the high-dimensional setting ($p=400$), and weaker signal effects in path-a ($a=0.2, 0.4$), the performance of most methods degrades substantially, while GKMS consistently achieves high power and reliable FDR control, further demonstrating its robustness to both dimensionality and weak signal conditions. 

Overall, our proposed method demonstrates robust and reliable performance across linear and nonlinear models, making it an effective tool for detecting true signals under complex scenarios, especially in high-dimensional, correlated data and weak signal conditions. Additional simulation results can be found in Appendix C.3.
}
\begin{figure}[!ht]
    \centering
     \includegraphics[scale=0.7]{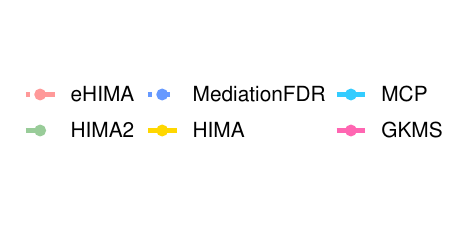}\\
    \includegraphics[scale=0.33]{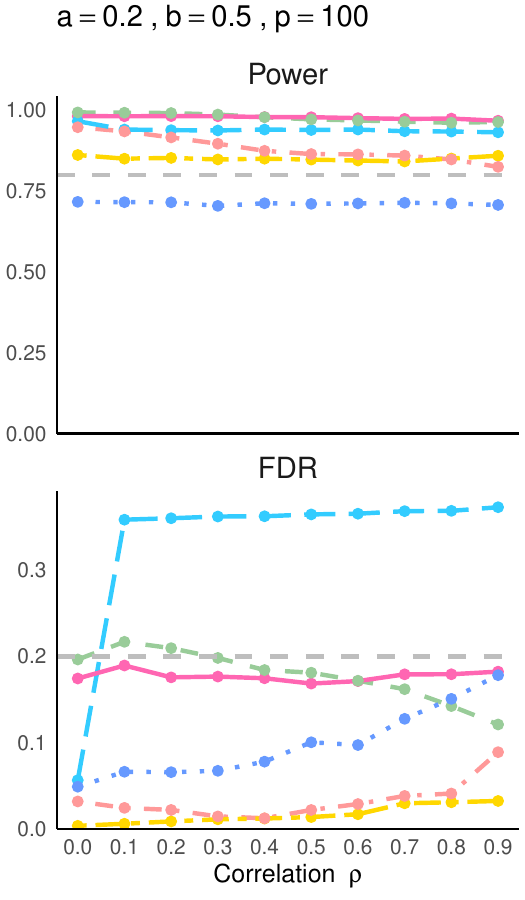}
    \includegraphics[scale=0.33]{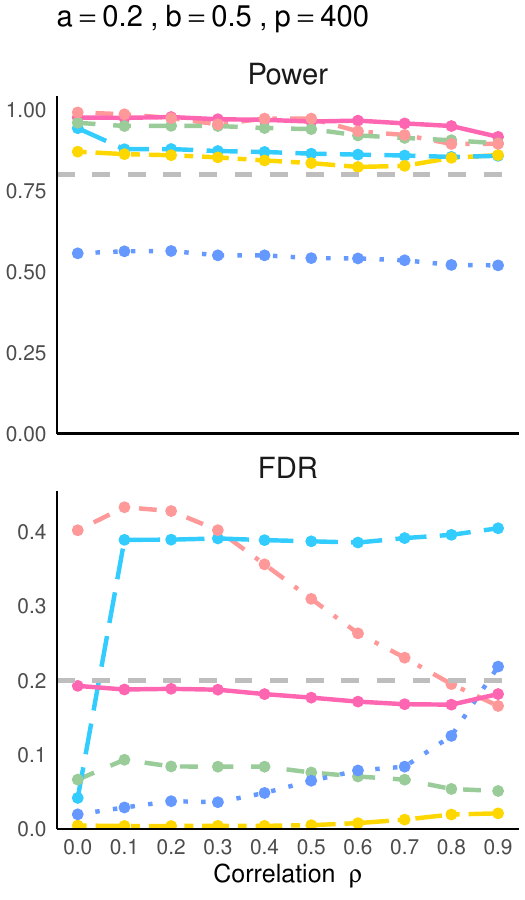}
    \includegraphics[scale=0.33]{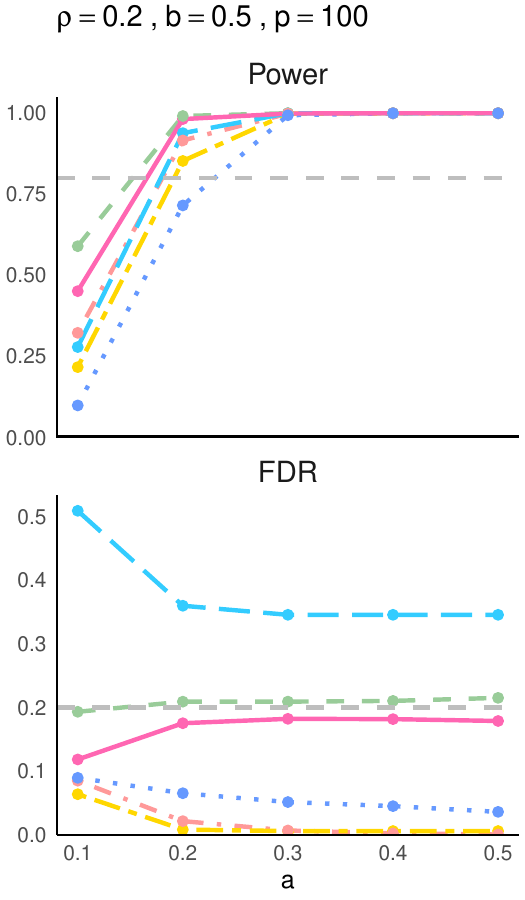}
    \includegraphics[scale=0.33]{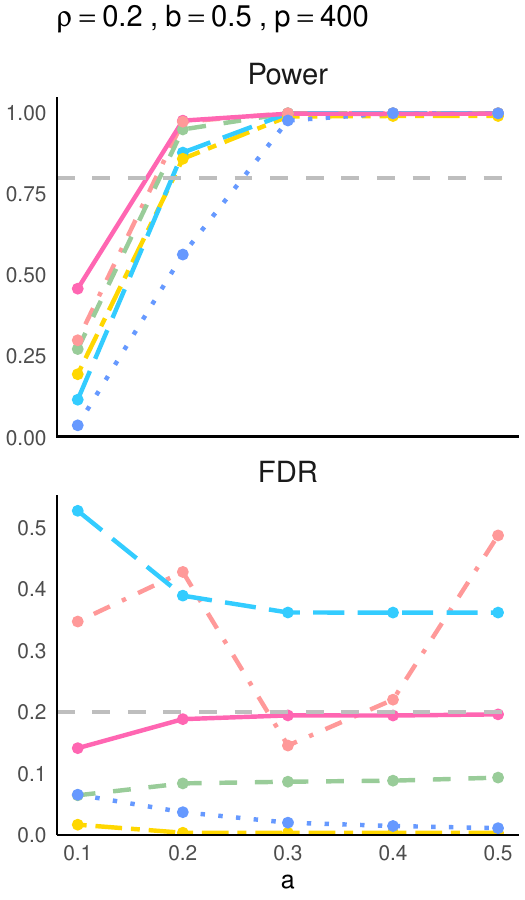}\\
    \includegraphics[scale=0.33]{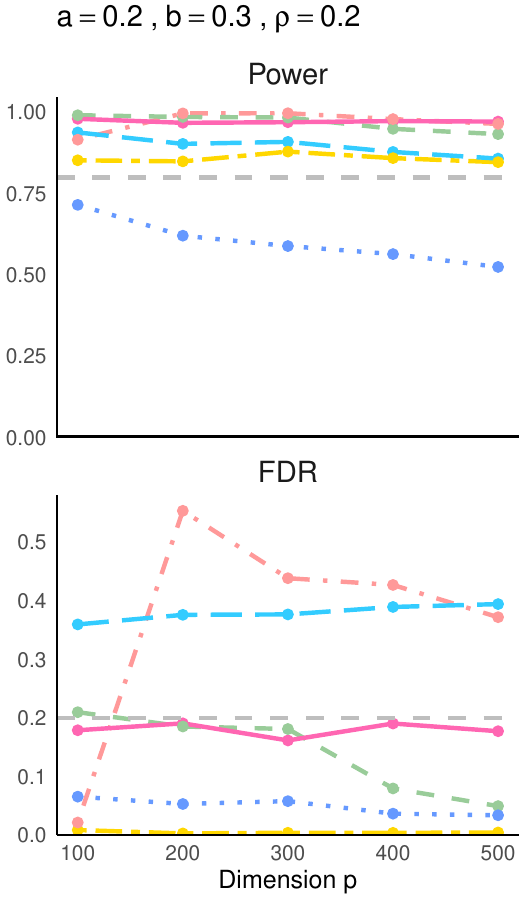}
    \includegraphics[scale=0.33]{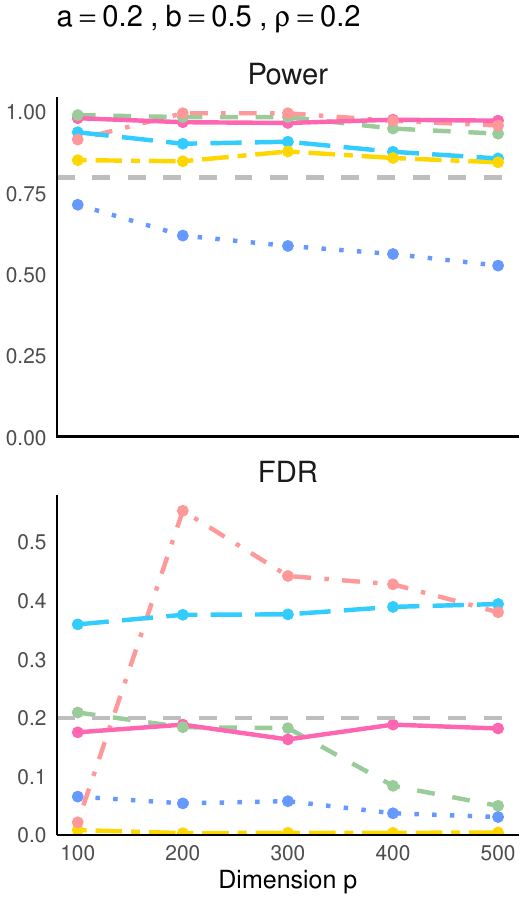}
    \includegraphics[scale=0.33] {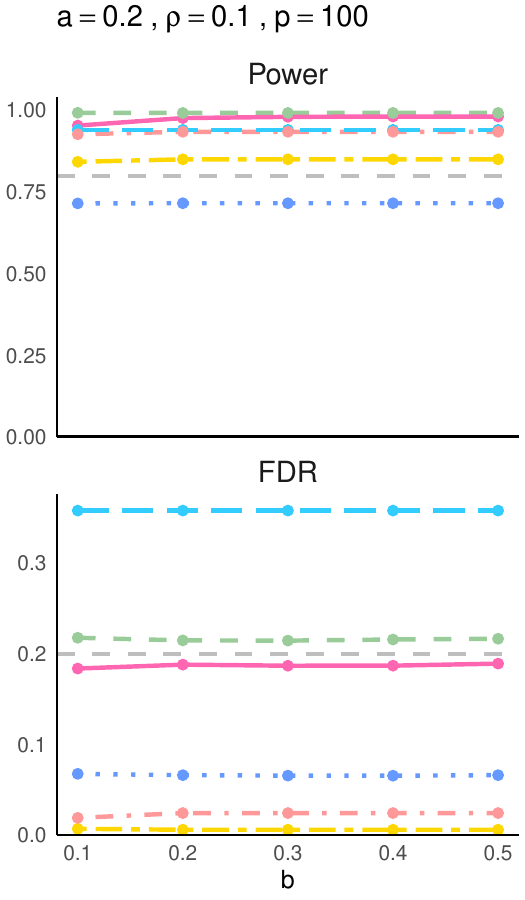}
    \includegraphics[scale=0.33]{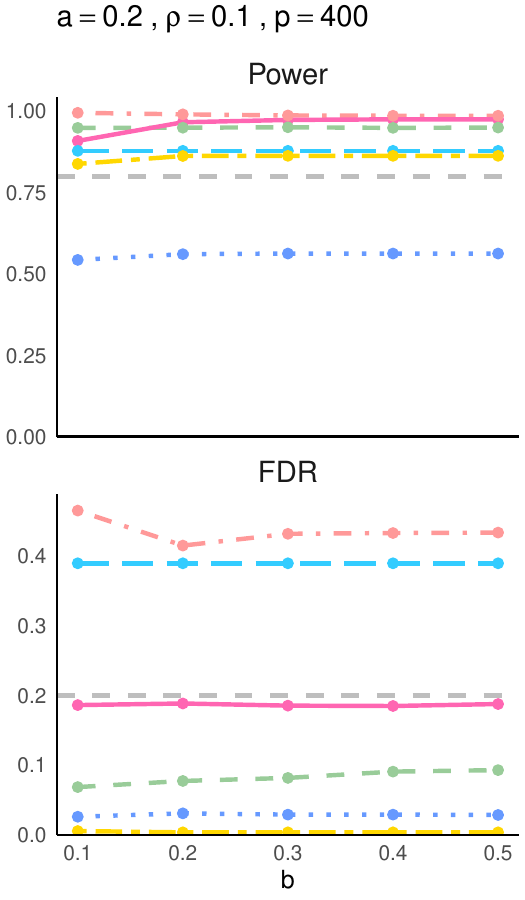}
   \caption{Simulation results for the empirical power and false discovery rate (FDR) with sample size $n=1000$ are shown for correlated mediators with a defined correlation structure across 100 replications. The results are stratified by the effect size of the correlation structure ($\rho$),
   the number of potential mediators (p), the effect size of path-a (a), and the effect size of path-b (b) in the linear models.}
   \label{fig:chg_linear}
\end{figure}

\begin{table}[ht]
    \caption{Simulation results for the empirical power and FDR of \textit{\textbf{Setting 1}} (nonlinear), with $n=1000$, $p=100$}
     \centering
    \begin{tabular}{ccc|cccccc}
    \hline
    \multicolumn{3}{c|}{} & \multicolumn{6}{c}{\textbf{FDR}} \\ \hline
   $\rho$ & $a$ & $b$ & HIMA & HIMA2 & eHIMA & MCP & MediationFDR & GKMS \\ \hline
   0.1   & 0.4 & 0.5   & 0.020 & 0.038 & 0.337 & 0.025 & 0.000 &\textbf{0.083} \\ 
   0.1   & 0.4 & 1   & 0.040 & 0.033 & 0.312 & 0.020 & 0.000 & \textbf{0.077} \\ 
   0.1   & 0.6 & 0.5   & 0.040 & 0.048 & 0.449 & 0.010 & 0.000 & \textbf{0.113} \\ 
   0.1  & 0.6 & 1   & 0.030 & 0.058 & 0.469 & 0.010 & 0.000& \textbf{0.115} \\    
   0.3 & 0.4 & 0.5   & 0.095 & 0.045 & 0.266 & 0.081 & 0.000 & \textbf{0.112} \\ 
   0.3 & 0.4 & 1   & 0.085 & 0.058 & 0.268 & 0.096 & 0.000 & \textbf{0.121} \\ 
   0.3 & 0.6 & 0.5  & 0.040 & 0.063 & 0.311 & 0.059& 0.000 & \textbf{0.111} \\ 
   0.3 & 0.6 & 1   & 0.030 & 0.058 & 0.325 & 0.058 & 0.000 & \textbf{0.111} \\ 
   0.5 & 0.4 & 0.5   & 0.095 & 0.071 & 0.307 & 0.102 &0.000 & \textbf{0.100} \\ 
   0.5 & 0.4 & 1  & 0.095 & 0.087 & 0.309 & 0.101 & 0.000 & \textbf{0.108} \\ 
   0.5 & 0.6 & 0.5   & 0.075 & 0.096& 0.347 & 0.093 & 0.000 & \textbf{0.097} \\ 
   0.5 & 0.6 & 1   & 0.070 & 0.105 & 0.332 & 0.088 & 0.000 & \textbf{0.085} \\ \hline
    \multicolumn{3}{c|}{} & \multicolumn{6}{c}{\textbf{Power}} \\ \hline
   $\rho$ & $a$ & $b$ & HIMA & HIMA2 & eHIMA & MCP & MediationFDR & GKMS \\ \hline
  0.1   & 0.4 & 0.5   & 0.004 & 0.013 & 0.044 & 0.006 & 0.000 & \textbf{0.754} \\ 
   0.1   & 0.4 & 1   & 0.003 & 0.012 & 0.045 & 0.006 & 0.000 & \textbf{0.754} \\ 
   0.1   & 0.6 & 0.5   & 0.003 & 0.020 & 0.035 & 0.005 & 0.000 & \textbf{0.862} \\ 
   0.1  & 0.6 & 1   & 0.003 & 0.020 & 0.034 & 0.004 & 0.000 & \textbf{0.868} \\    
   0.3 & 0.4 & 0.5   & 0.005 & 0.017 & 0.058& 0.089 & 0.000 & \textbf{0.812} \\ 
   0.3 & 0.4 & 1   & 0.005 & 0.018& 0.059 & 0.088& 0.000 & \textbf{0.804} \\ 
   0.3 & 0.6 & 0.5  & 0.001 & 0.011 & 0.057 & 0.068 & 0.000 & \textbf{0.841} \\ 
   0.3 & 0.6 & 1   & 0.001 & 0.011 & 0.059 & 0.068 & 0.000 & \textbf{0.837} \\ 
   0.5 & 0.4 & 0.5   & 0.005 & 0.031 & 0.047 & 0.166 & 0.000 & \textbf{0.795} \\ 
   0.5 & 0.4 & 1  & 0.006 & 0.033 & 0.048 & 0.168 & 0.000 & \textbf{0.788} \\ 
   0.5 & 0.6 & 0.5   & 0.005 & 0.032& 0.040 & 0.140 & 0.000 & \textbf{0.809} \\ 
   0.5 & 0.6 & 1   & 0.005 & 0.034 & 0.043 & 0.137 & 0.000 &\textbf{0.776} \\ \hline
    \end{tabular}
    
    \label{tab:sim2}
\end{table}
\begin{table}[ht]
    \caption{Simulation results for the empirical power and FDR of \textit{\textbf{Setting 2}} (nonlinear), with $n=1000$, $p=400$}
     \centering
   
    \begin{tabular}{ccc|cccccc}
    \hline
    \multicolumn{3}{c|}{} & \multicolumn{6}{c}{\textbf{FDR}} \\ \hline
   $\rho$ & $a$ & $b$ & HIMA & HIMA2 & eHIMA & MCP & MediationFDR & GKMS \\ \hline
   0   & 0.2 & 0.5 & 0.020 & 0.000 & 0.418 & 0.010 & 0.000 & \textbf{0.162} \\
   0   & 0.2 & 1.0 & 0.000 & 0.005 & 0.385 & 0.005 & 0.000 & \textbf{0.150} \\
   0   & 0.5 & 0.5 & 0.000 & 0.000 & 0.451 & 0.000 & 0.000 & \textbf{0.139} \\
   0   & 0.5 & 1.0 & 0.000 & 0.000 & 0.438 & 0.000 & 0.000 & \textbf{0.148} \\
   0.3 & 0.2 & 0.5 & 0.050 & 0.065 & 0.546 & 0.059 & 0.000 & \textbf{0.124} \\
   0.3 & 0.2 & 1.0 & 0.050 & 0.062 & 0.513 & 0.051 & 0.000 & \textbf{0.116} \\
   0.3 & 0.5 & 0.5 & 0.070 & 0.043 & 0.608 & 0.040 & 0.000 & \textbf{0.102} \\
   0.3 & 0.5 & 1.0 & 0.070 & 0.053 & 0.551 & 0.037 & 0.000 & \textbf{0.112} \\ \hline
    \multicolumn{3}{c|}{} & \multicolumn{6}{c}{\textbf{Power}} \\ \hline
   $\rho$ & $a$ & $b$ & HIMA & HIMA2 & eHIMA & MCP & MediationFDR & GKMS \\ \hline
   0   & 0.2 & 0.5 & 0.002 & 0.001 & 0.051 & 0.001 & 0.000 &\textbf{0.833} \\
   0   & 0.2 & 1.0 & 0.001 & 0.002 & 0.050 & 0.003 & 0.000 &\textbf{0.858} \\
   0   & 0.5 & 0.5 & 0.001 & 0.001 & 0.051 & 0.001 & 0.000 &\textbf{0.881} \\
   0   & 0.5 & 1.0 & 0.001 & 0.001 & 0.048 & 0.002 & 0.000 &\textbf{0.900} \\
   0.3 & 0.2 & 0.5 & 0.001 & 0.002 & 0.117 & 0.051 & 0.000 &\textbf{0.706} \\
   0.3 & 0.2 & 1.0 & 0.001 & 0.002 & 0.130 & 0.053 & 0.000 &\textbf{0.729} \\
   0.3 & 0.5 & 0.5 & 0.001 & 0.004 & 0.105 & 0.051 & 0.000 &\textbf{0.618} \\
   0.3 & 0.5 & 1.0 & 0.001 & 0.003 & 0.121 & 0.052 & 0.000 &\textbf{0.655} \\ \hline
    \end{tabular}
    
    \label{tab:sim3}
\end{table}
\section{Data Analysis} \label{sec:analysis}
The Alzheimer’s Disease Neuroimaging Initiative (ADNI) is a large-scale longitudinal study designed to identify biomarkers for the early detection and progression of Alzheimer’s disease (AD) \citep{jack2008alzheimer}. 
Evidence suggests that the development and progression of AD differ between men and women, but the mechanisms behind the gender differences remain poorly understood \citep{pike2017sex}. Structural MRI has been established as a robust biomarker for assessing the stage and severity of neurodegeneration in AD \citep{vemuri2010role}. {\color{black}In this study, we focus on 1097 mild cognitive impairment (MCI) patients (58\% Male, 42\% Female) and aim to identify structural MRI-derived measures in the baseline that mediate the relationship between gender and changes in cognitive performance.}

{\color{black}To obtain the structural MRI data, we first convert ADNI dataset to follow the Brain Imaging Data Structure (BIDS) rule via Clinica \citep{routier2021clinica} and applied fMRIPrep \citep{esteban2019fmriprep} with the FreeSurfer \citep{fischl2012freesurfer} reconstruction workflow to pre-process the data. This pipeline outputs 555 structural MRI-derived measures of each subjects (e.g. cortical thickness, surface area and brain volumes) corresponding to specific brain regions or morphometric features, such as the average cortical thickness in various gyri and sulci, and volumetric data for subcortical structures.  The full list of the 555 structural MRI-derived measures can be found in Web Appendix D.1.  
}


{\color{black}The outcome of interest in this study is the change in multiple cognitive assessment scores from baseline to 17 years of follow-up. These assessments are typically conducted at baseline and repeated at regular intervals (e.g., every 3 or 6 months). To account for the missing visits, we apply the last observation carried forward (LOCF) method \citep{lachin2016fallacies} to impute missing cognitive scores. To evaluate different aspects of cognitive function, we apply multiple cognitive scores, including the Clinical Dementia Rating Sum of Boxes (CDR-SB), Alzheimer’s Disease Assessment Scale-Cognitive Subscale (ADAS-Cog) and Mini-Mental State Examination (MMSE). Higher CDR-SB and ADAS-Cog scores, and lower MMSE scores, each indicate greater AD severity. More details about those cognitive scores can be found in Appendix D.2.


We apply the proposed GKMS method to identify important structural MRI-derived measures from a set of $p=555$ candidate correlated mediators that may mediate the relationship between gender and changes in cognitive scores among $n=1097$ MCI patients. 
} We consider Age at baseline as a confounding variable $V$. The FDR is set as $q=0.2$. For exposure X, we set males as the reference group. $\Y$ are the changes from baselines in the cognitive assessment scores CDR-SB, ADAS-Cog and MMSE. Both the outcome $\Y$ and the mediators $\M$ are scaled, ensuring that the values the $\alpha$ and $\beta$ are the standardized regression coefficients. {\color{black}To capture the interactions among the mediators, we applied our proposed GKMS method with Random Forests in path-b. Table \ref{tab:GKMS} summarizes the mediator selection results with their estimated mediation effects for changes in CDR-SB, ADAS-Cog, and MMSE scores. The mediator effects were estimated using the g-formula with numerical Monte Carlo integration as described by \citet{imai2010general}. More details on the estimation of mediation effects are provided in Appendix D.3. 

Our proposed method identified both shared and distinct mediators across the three cognitive outcome measures. The hippocampal and amygdala volumes were consistently selected as important mediators for all three cognitive outcomes. Cortical thickness in the left angular gyrus was chosen as a significant mediator for both CDR-SB and ADAS-Cog. Furthermore, the cortical thickness in the middle frontal gyrus (left) was also selected as a positive mediator for MMSE.


As a sensitivity analysis, we applied our proposed method with Lasso linear regression in path-b (GKMS (Lasso)), HIMA, HIMA2, eHIMA, and MediationFDR methods for comparison. The mediation effects were estimated using structural equation modeling (SEM) through the \texttt{lavaan} R package \citep{rosseel_lavaan_2012}. Tables 4, 5, and 6 summarize the results for CDR-SB, ADAS-Cog, and MMSE scores, respectively. Both hippocampal and amygdala volumes are commonly selected in those methods. More details are provided in Appendix D.4. 
}

{\color{black} The results indicate that hippocampal atrophy mediates sex differences in Alzheimer's disease progression. Specifically, women exhibit faster hippocampal volume loss in structural brain changes, which contribute to more severe clinical symptoms and accelerated cognitive decline. Furthermore, the volume of the amygdala significantly mediates the relationship between sex and cognitive decline in Alzheimer's disease. Specifically, females exhibited smaller amygdala volumes than males, which were in turn associated with greater cognitive impairment.}


{\color{black}These findings are consistent with previous independent studies. For example, 
a longitudinal analysis found that women with probable Alzheimer's disease exhibited significantly faster hippocampal atrophy compared to men \citep{ardekani2016analysis}. This accelerated atrophy was associated with a more rapid decline into dementia, highlighting female sex as a risk factor for faster disease progression. Another study analyzing resting-state fMRI data from individuals with mild cognitive impairment (MCI) showed that females had weaker hippocampal connectivity to the precuneus cortex and brainstem compared to males. This diminished connectivity may contribute to the increased susceptibility of women to cognitive decline and progression to Alzheimer's disease \citep{williamson2022sex}.
} \cite{poulin2011amygdala} reported that amygdala atrophy is associated with increased dementia severity. Males consistently exhibit larger amygdala volumes than females, after adjusting for total brain volume \citep{goldstein2001normal}. {\color{black}Moreover, Tau-PET studies have shown that tau accumulation in the amygdala surpasses that in the entorhinal cortex, indicating the amygdala may be among the first brain regions affected by tau pathology in Alzheimer's disease \citep{stouffer2024amidst,lupton2016effect}.} 
Our analysis provides insights into how gender influences cognitive decline and contributes to a better understanding of gender-specific differences in dementia progression.

\section{Discussions} 
In this work, we present a knockoff-based mediator selection procedure (GKMS) with FDR control guarantees. We demonstrated its finite sample performance using extensive simulation settings and implemented it to identify mediators from structural MRI-derived measures for the relationship between gender and changes in cognitive performance in the MCI patients. Our proposed method successfully controls FDR in all the experimented scenarios, including linear, nonlinear settings, with correlated mediators and interactions among mediators.

{\color{black}In some extreme model mis-specified scenarios, the existing mediator selection methods we compared with (e.g., HIMA, HIMA2, eHIMA, MCP, and MediationFDR) exhibit near-zero power, failing to detect any mediators. This highlights a key limitation of their reliance on the model assumptions. In contrast, our proposed ``model-free" method, GKMS remain robust, consistently identifying meaningful mediators while maintaining FDR control.}

The proposed method is a modest extension from the Simultaneous Knockoff procedure, where we treat path-a and path-b as two experiments and we simultaneously control the FDR for both. With linear simulation settings, other existing mediator selection methods have similar performances as our proposed GKMS method. However, when interactions exist, or the true path-b model is nonlinear, some existing methods (eHIMA, MCP, HIMA2) fail to control FDR, MediationFDR and HIMA have suboptimal power performances. Since complex correlation structures tend to exist in brain imaging data and the mediation mechanisms are potentially complicated, the proposed GKMS method is perferred in our application example. 

For future work, the GKMS framework can be further extended for group-level, or multi-modal mediator selections by using the group Knockoff constructions \cite{dai2016knockoff}, and for mediator selection based on longitudinally measured mediators and outcomes (e.g. the ADNI data). For different real data scenarios, the different mediator selection methods demonstrate different performances; a data-driven approach to make decision on the selection of the mediator selection method to be adopted is also worth exploring.


\begin{table}[!ht]
\caption{ADNI study mediator selection results using \text{GKMS} method for the three outcomes: \textbf{CDR-SB}, \textbf{ADAS-Cog}, and \textbf{MMSE}.(L): Left; (R): Right.} 
\label{tab:GKMS} 
\centering 
\begin{threeparttable}
\begingroup
  \renewcommand{\arraystretch}{0.9}      
  \setlength{\tabcolsep}{5pt}            
  \fontsize{8}{9.5}\selectfont           
\begin{tabular}{l l l} 
\hline 
Cognitive Scores & Selected Mediator & Effect (95\% CI; p-value)\textsuperscript{*} \\ [0.5ex] 
\hline 
\textbf{CDR-SB}
 &  Volume of Hippocampus (L) &  {\color{black}$\alpha$ = -0.52 (-0.627, -0.414; $<$0.001)} \\ 
 & & {\color{black}$\beta$ = -0.157 (-0.388, -0.228; $<$0.001))} \\ 
 & & {\color{black}Mediation = 0.063 (0.013, 0.092; $<$0.001)} \\
 & &\\
 &  Volume of Amygdala (L) & {\color{black}$\alpha$ = -0.505 (-0.614, -0.397; $<$0.001)} \\ 
 & & {\color{black}$\beta$ = -0.213 (-0.458, -0.293;$<$0.001)} \\ 
 & & {\color{black} Mediation =0.092 (0.051, 0.137; $<$0.001)} \\
 &&\\
  & Volume of Hippocampus (R) &  {\color{black}$\alpha$ = -0.551 (-0.658, -0.444; $<$0.001)} \\ 
 & & {\color{black}$\beta$ = -0.162 (-0.393, -0.236; $<$0.001) }\\ 
 & & {\color{black}Mediation = 0.026 (0.004, 0.072; 0.02)} \\
 & & \\
& ThickAvg\textsuperscript{1} of Angular Gyrus (L) & {\color{black}$\alpha$ = 0.23 (0.115, 0.346; $<$0.001)} \\ 
 & & {\color{black}$\beta$ = -0.135 (-0.342, -0.201; $<$0.001)} \\ 
 & & {\color{black}Mediation = -0.024 (-0.045, -0.005; $<$0.001)} \\
 & &\\
 &  ThickAvg of Sulcus (L) &{\color{black}$\alpha$ = 0.192 (0.082, 0.302; $<$0.001)} \\ 
 & & {\color{black}$\beta$ = -0.165 (-0.405, -0.234; $<$0.001)} \\ 
 & & {\color{black}Mediation =-0.016 (-0.041, 0.003; 0.08) }\\
 \hline 
\textbf{ADAS-Cog}
 & Volume of Hippocampus (L) &  {\color{black}$\alpha$ = -0.52 (-0.627, -0.414; $<$0.001)} \\ 
 & & {\color{black}$\beta$ = -0.118 (-0.364, -0.186; $<$0.001)} \\ 
 & & {\color{black}Mediation = 0.041 (0.01, 0.073; 0.016)} \\
 & & \\
&  Volume of Amygdala (L) & {\color{black}$\alpha$ = -0.505 (-0.614, -0.397; $<$0.001)} \\ 
 & & {\color{black}$\beta$ = -0.171 (-0.411, -0.249; $<$0.001)} \\ 
 & & {\color{black}Mediation = 0.102 (0.039, 0.127; $<$0.001) }\\
 & &\\
&  ThickAvg of Angular Gyrus (L) & {\color{black}$\alpha$ = 0.233 (0.118, 0.349; $<$0.001)} \\ 
 & & {\color{black}$\beta$ = -0.164 (-0.418, -0.222 $<$0.001)} \\ 
 & & {\color{black}Mediation = -0.026 (-0.057, -0.005;  0.02)} \\
 & &\\
 &   ThickAvg of S.temporal.sup\textsuperscript{2} (R) & {\color{black}$\alpha$ = 0.092 (-0.023, 0.207; 0.117)} \\ 
 & & {\color{black}$\beta$ = -0.177 (-0.422, -0.235;  $<$0.001)} \\ 
 & & {\color{black}Mediation =-0.004 (-0.032, 0.008; 0.22) }\\
 & & \\
  &   Volume of white matter in ITL\textsuperscript{3} (L) & {\color{black}$\alpha$ = -0.895 (-1.004, -0.785; $<$0.001)} \\ 
 & & {\color{black}$\beta$ = 0.006 (-0.155, 0.185; 0.9) }\\ 
 & & {\color{black}Mediation =-0.026 (-0.057, 0.022; 0.34)} \\
\hline 
\textbf{MMSE}
 & Volume of Hippocampus (L) &  {\color{black}$\alpha$ = -0.52 (-0.627, -0.414; $<$0.001)} \\ 
 & & {\color{black}$\beta$ = 0.129 (0.184, 0.375; $<$0.001)} \\ 
 & & {\color{black}Mediation = -0.063 (-0.109, -0.016;  0.001)} \\
 & &\\
&  Volume of Amygdala (L) & {\color{black}$\alpha$ = -0.505 (-0.614, -0.397;  $<$0.001)} \\ 
 & & {\color{black}$\beta$ =0.165 (0.256, 0.409; $<$0.001) }\\ 
 & & {\color{black}Mediation = -0.077 (-0.142, -0.026; $<$0.001)} \\
 & & \\
&  Volume of Amygdala (R) & {\color{black}$\alpha$ = -0.596 (-0.706, -0.487; $<$0.001)} \\ 
 & & {\color{black}$\beta$ = 0.157 (0.224, 0.396; $<$0.001)} \\ 
 & & {\color{black}Mediation =-0.036 (-0.075, 0.001; 0.04) }\\
 & & \\
 &  Volume of Hippocampus (R) & {\color{black}$\alpha$ = -0.551 (-0.658, -0.444; $<$0.001)} \\ 
 & & {\color{black}$\beta$ = 0.132 (0.2, 0.366; $<$0.001)}\\
 & & {\color{black}Mediation = -0.076 (-0.122, -0.029; 0.032)} \\
 & & \\
&  ThickAvg of G-front-middle\textsuperscript{4} (L) & {\color{black}$\alpha$ = 0.29 (0.175, 0.404; $<$0.001)} \\ 
 & & {\color{black}$\beta$ = 0.136 (0.206, 0.392; $<$0.001)} \\ 
 & & {\color{black}Mediation = 0.028 (0.013, 0.056; $<$0.001)} \\
\hline 
\end{tabular}

\begin{tablenotes}
        \item {\color{black}\textsuperscript{*}$\beta$ effects are the importance scores with direction derived from partial dependence analysis comparing outcome predictions at the first and third quartiles of mediator values. Mediation effects are estimated using Monte Carlo integration as described by \citet{imai2010general}, and confidence intervals are obtained via bootstrap.}     
        \item Abbreviations (all extracted from the aparc.a2009s statistics file):
        \item \textsuperscript{1} ThickAvg: Average cortical thickness of the specified cortical region.
        \item \textsuperscript{2} S.temporal.sup: Superior temporal sulcus region.
        \item \textsuperscript{3} ITL: Inferior temporal lobe, denoting white matter volume in the left hemisphere of the inferior temporal lobe.
        \item  \textsuperscript{4} G-front-middle: Middle frontal gyrus, a cortical brain region.
     \end{tablenotes}
\endgroup
  \end{threeparttable}
\end{table}

\section*{Data Availability}
The ADNI data that support the findings in this paper can be accessed through the National Institute of Mental Health Data Archive (\url{https://adni.loni.usc.edu/data-samples/adni-data/}). The R code is available at \url{https://github.com/RunqiuWang22/Generalized_Knockoff_Mediator_Selection}.

\bibliography{example_paper}
\bibliographystyle{icml2022}

\pagebreak
\section*{Web Appendix A: Details for the choices of the functions in Mediator Selection Algorithm}
\subsection*{Appendix A.1: Model-X Knockoff assumptions and construction methods}

The general sampling method as described in \citet{candes2018} can be applied to all kinds of data types (continuous, categorical, count, or mixed with or without missing data) as long as the joint distribution of covariate $\X$ is known or can be estimated. However, for the implementation, their current R package only allows for continuous $\X$, so in simulation, we only consider the Model-X Knockoff construction methods for continuous $\X$. Specifically, two Model-X Knockoff construction methods are reviewed below:
\begin{itemize}
\item Gaussian: When the distribution of $\X$ is assumed to be Gaussian with the variance-covariance matrix $\bmSigma$, we can sample $\widetilde{\X}$ from $\widetilde{\X}|\X\sim N(\bmmu, \mathbf{V})$, where $\bmmu$ and $\mathbf{V}$ are given by,
\begin{eqnarray*}
\bmmu&=&\X-\X\bmSigma^{-1}\diag\{\bms\},\\
\mathbf{V}&=&2\diag\{\bms\}-\diag\{\bms\}\bmSigma^{-1}\diag\{\bms\}.
\end{eqnarray*}
We use the R function \textit{create.gaussian} within the R package \textit{knockoff} to implement this construction method.
\item Second order: The second-order Model-X knockoff construction method tries to sample $\widetilde{\X}$ such that 
\begin{equation*}
Cov([\X,\widetilde{\X}]) = \left(\begin{array}{cc}\bmSigma&\bmSigma-\diag\{\bms\}\\\bmSigma-\diag\{\bms\}&\bmSigma\end{array}\right).\end{equation*} 
where the requirements of $\bms$ are the same as for the Fixed-X knockoff and can be solved using the approximate semidefinite program (ASDP) algorithm as given in \citep{candes2018}. We use the R function \textit{create.second} within the R package \textit{knockoff} to implement this construction method.
\end{itemize}
The knockoffs $\Xp$ can also be generated using various advanced algorithms \citep{romano2019,liu2019,bates2020,spector2020}. 

\subsection*{Appendix A.2: Statistics compatible with each knockoff construction method}
The statistics described in this section are for individual hypothesis testing for simplicity of notation. It can be extended to statistics for the group hypothesis testing similar to those used in \citet{dai2016knockoff}.
\subsection*{Appendix A.2.1: Statistics compatible with the Fixed-X knockoff construction method}
\begin{itemize}
\item Order of selection (Lasso): We can choose the lasso variable selection procedure and construct the statistics as $\widehat{\bmbeta}(\lambda)$, where 
\[\widehat{\bmbeta}(\lambda) = \arg\min_{\bmb\in\R^{2p}} \left\{\norm{\Y - [ \X \ \Xp ] \bmb}_2^2 +\lambda \norm{\bmb}_1\right\}.\]
We can run over a range of $\lambda$ values decreasing from $+\infty$ (a fully sparse model) to $0$ (a fully dense model) and define $Z_j$ as the maximum $\lambda$ such that $\widehat{\beta}_j(\lambda)\neq 0$. If there is no $\lambda$ such that $\widehat{\bmbeta}_j(\lambda)\neq 0$, then we will simply define $Z_j$ as 0.
\item Absolute coefficient (Lasso): We can use $|\widehat{\beta}_j(\lambda)|$ as defined above with a specific $\lambda$ value. 
\end{itemize}

\subsection*{Appendix A.2.2: Statistics compatible with the Model-X knockoff construction method}
For the Model-X knockoffs, very general conditional models (such as generalized linear models or nonlinear models) can be used. In addition to the statistics listed in Appendix A.2.1, we can also use the following statistics.

\begin{itemize}
\item Absolute coefficient (glmnet): We can use $|\widehat{\beta}_j(\lambda)|$ from the penalized generalized linear regression or the penalized Cox regression model of $Y$ on $[\X\  \widetilde{\X}]$ with either a specific $\lambda$ value or a $\widehat{\lambda}$ estimated from cross-validation. 
\item Standardized coefficient (glmnet): We can also use the standardized regression coefficients $|\widehat{\beta}_j(\lambda)|/\widehat{\text{SE}}(\widehat{\beta}_j(\lambda))$ from the penalized generalized linear regression or the penalized Cox regression model of $Y$ on $[\X\   \widetilde{\X}]$ with either a specific $\lambda$ value or a $\widehat{\lambda}$ estimated from cross-validation.
\item Order of selection (glmnet): We can use the minimum $\lambda$ such that the regression coefficient becomes 0 from the penalized generalized linear regression or the penalized Cox regression model of $Y$ on $[\X\ \widetilde{\X}]$, or the reciprocal of the order of each variable to be included in the model when increasing the number of variables to be selected.
\item Variable importance factor: We can use the variable importance factors from the random forest fitting of $Y$ on $[\X\ \widetilde{\X}]$ with either fixed tuning parameters or tuning parameters selected from cross-validation. 
\end{itemize}       

\renewcommand{\thelemma}{B\arabic{lemma}}
\section*{Web Appendix B: Proofs}
\subsection*{B.1: Technical Lemmas} 
\newtheorem{mainlemma}{Lemma}
\begin{mainlemma}\label{lem:0}\label{Lamma:Lemb1}
Let $\W = f([\Z^a,\tilde{\Z}^a],[\Z^b,\tilde{\Z}^b])$ where $f$ is an OSFF. If $[Z^a_j, \tilde{Z}^a_j] \eqd [\tilde{Z}^a_j, Z^a_j]$ for all $j\in \mathcal{H}_a$, $[Z^b_j, \tilde{Z}^b_j] \eqd [\tilde{Z}^b_j, Z^b_j]$ for all $j\in \mathcal{H}_b$, and $[\Z^a,\tilde{\Z}^a]$ and $[\Z^b,\tilde{\Z}^b]$ are independent to each other, then for all $j \in \mathcal{H}$ and $\epsilon_j \sim \{\pm 1\}$ for all $j \in \mathcal{H}$. Then $W_j \eqd W_j \cdot \epsilon_j$. 
\end{mainlemma}
\begin{proof}
For any $j\in \mathcal{H}$, if $j\in \mathcal{H}_a$, then we have
$[Z^a_j, \tilde{Z}^a_j] \eqd [\tilde{Z}^a_j, Z^a_j]$. Since $f$ is OSFF, we have $-W_j=f([\Z^a,\tilde{\Z}^a]_{Swap\{j\}},[\Z^b,\tilde{\Z}^b])_j$. Since $[\Z^a,\tilde{\Z}^a]$ and $[\Z^b,\tilde{\Z}^b]$ are independent, and $[Z^a_j, \tilde{Z}^a_j] \eqd [\tilde{Z}^a_j, Z^a_j]$, we have $$f([\Z^a,\tilde{\Z}^a]_{Swap\{j\}},[\Z^b,\tilde{\Z}^b])_j\eqd f([\Z^a,\tilde{\Z}^a],[\Z^b,\tilde{\Z}^b])_j.$$ Together, we have $W_j\eqd -W_j$, and thus we have $W_j\eqd W_j\cdot\epsilon_j$. Otherwise, if $j\in \mathcal{H}_b$, then we have
$[Z^b_j, \tilde{Z}^b_j] \eqd [\tilde{Z}^b_j, Z^b_j]$. Given $f$ is OSFF, we have $-W_j=f([\Z^a,\tilde{\Z}^a],[\Z^b,\tilde{\Z}^b]_{Swap\{j\}})_j$. Since $[\Z^a,\tilde{\Z}^a]$ and $[\Z^b,\tilde{\Z}^b]$ are independent, and $[Z^b_j, \tilde{Z}^b_j] \eqd [\tilde{Z}^b_j, Z^b_j]$, we have $$f([\Z^a,\tilde{\Z}^a],[\Z^b,\tilde{\Z}^b]_{Swap\{j\}})_j\eqd f([\Z^a,\tilde{\Z}^a],[\Z^b,\tilde{\Z}^b])_j.$$ Together, we also have $W_j\eqd -W_j$, and thus we have $W_j\eqd W_j \cdot \epsilon_j$. So we finish the proof that for any $j\in \mathcal{H}$, we have $W_j\eqd W_j \cdot \epsilon_j$.
\end{proof}

\begin{mainlemma}\label{Lamma:Lemb2}
Assume $p_j$s are weakly correlated and exchangeable and $p_j\geq \text{Uniform}[0,1]$. For $k=m,m-1,\cdots,1,0$, put $V^{+}(k)=\#\{j:1\leq j\leq k, p_j\leq 1/2, j\in \mathcal{H}\}$ and $V^{-}(k)=\#\{ j:1\leq j\leq k, p_j> 1/2, j\in \mathcal{H}\}$ with the convention that $V^{\pm}(0)=0$. Let $\mathcal{F}_k$ be the filtration defined by knowing all the non-null $p$-values, as well as $V^{\pm}(k')$ for all $k'\geq k$. Then the process $M(k)=\frac{V^{+}(k)}{1+V^{-}(k)}$ is a super-martingale running backward in time with respect to $\mathcal{F}_k$. For any fixed $q$, $\widehat{k}=\widehat{k}_{+}$ or $\widehat{k}=\widehat{k}_{0}$ as defined in the proof of theorem \ref{pf:thm1} are stopping times, and as consequencis
\begin{eqnarray*}
\EE{\frac{\#\{j\leq \widehat{k}:p_j\leq 1/2, j\in \mathcal{H}\}}{1+\#\{j\leq \widehat{k}:p_j>1/2, j\in \mathcal{H}\}}}\leq 1
\end{eqnarray*}
\end{mainlemma}

\begin{proof}
The first part of the proof follows exactly the same as \cite{barber2015}. The filtration $\mathcal{F}_k$ contains the information of whether $k$ is null and non-null process is known exactly. If $k$ is non-null, then $M(k-1)=M(k)$ and if $k$ is null, we have
\begin{eqnarray*}
M(k-1)=\frac{V^{+}(k)-\mathbbm{1}_{p_k\leq 1/2}}{1+V^{-}(k)-(1-\mathbbm{1}_{p_k\leq 1/2})}=\frac{V^{+}(k)-\mathbbm{1}_{p_k\leq 1/2}}{\left(V^{-}(k)+\mathbbm{1}_{p_k\leq 1/2}\right)\vee 1}
\end{eqnarray*}
Given that nulls are exchangeable, we have
\begin{eqnarray*}
\PP{\mathbbm{1}_{p_k\leq 1/2}|\mathcal{F}_k}=\frac{V^{+}(k)}{\left(V^{+}(k)+V^{-}(k)\right)}.
\end{eqnarray*}
So when $k$ is null, we have

\begin{eqnarray*}
&&\EE{M(k-1)|\mathcal{F}_k}\\
&=&\frac{1}{V^{+}(k)+V^{-}(k)}\left[V^{+}(k)\frac{V^{+}(k)-1}{V^{-}(k)+1}+V^{-}(k)\frac{V^{+}(k)}{V^{-}(k)\vee 1}\right]\\
&=&\frac{V^{+}(k)}{1+V^{-}(k)}\mathbbm{1}_{V^{-}(k)>0}+(V^{+}(k)-1)\mathbbm{1}_{V^{-}(k)=0}\\
&\leq &M(k)
\end{eqnarray*}

This finishes the proof of super-martingale property. $\widehat{k}$ is a stopping time with respect to $\{\mathcal{F}_k\}$ since $\{\widehat{k}\geq k\}\in \mathcal{F}_k$. So we have $\EE{M(\widehat{k})}\leq \EE{M(m)}=\EE{\frac{\#\{j:p_j\leq 1/2, j\in \mathcal{H}\}}{1+\#\{j:p_j>1/2, j\in \mathcal{H}\}}}$.

Let $X=\#\{j:p_j\leq 1/2, j\in \mathcal{H}\}$, since we don't have independence between the null $p_j$s, we cannot directly obtain $X\leq Binomial (m,1/2)$ and get the exact expectation control under 1. We rely on the weak correlation for asymptotic control when the number of features goes to infinity. Speicifically, given that $p_j\geq \text{Uniform}[0,1]$ and the weak correlation between $p_j$s, based on weak law of large numbers, we have $X/N\rightarrow_p c$, where $c=N^{-1}\sum_{i=1}^N Pr(p_j\leq 1/2)\leq 1/2$. Since the function $f(X)=\frac{X}{1+N-X}$ is bounded, we have $\EE{\frac{X}{1+N-X}}\rightarrow \frac{c}{1-c}$ and since the function $g(x)=\frac{x}{1-x}$ is non-decreasing, we have $\EE{\frac{X}{1+N-X}}\leq 1$ when $N\rightarrow \infty$.
\end{proof}

\begin{mainlemma}\label{Lamma:Lemb3}
Assume $p_j$s satisfy the requirement in Lemma \ref{Lamma:Lemb2}, then with the Knockoff filter, the mFDR can be controlled at level $q$ and with the Knockoff+ filter, the FDR can be controlled at level $q$ when $p\rightarrow \infty$.
\end{mainlemma}
\begin{proof}
The proof follows \cite{barber2015}. Let $m=\#\{j:W_j\neq 0\}$ and assume without loss of generality that $|W_1|\geq |W_2|\geq\cdots\geq |W_m|>0$. 

We first show the result for the knockoff+ threshold. Define $V=\#\{j\leq \widehat{k}_{+}: p_j\leq 1/2, j\in \mathcal{H}\}$ and $R=\#\{j\leq \widehat{k}_{+}:p_j\leq 1/2\}$ where $\widehat{k}_{+}$ satisfy that $|W_{\widehat{k}_{+}}|=\tau_{+}$ where $\tau_{+}$ is defined in Theorem 1, we have

\begin{align*}
\mathbb{E} \left[ \frac{V}{R \vee 1} \right]
&= \mathbb{E} \left[ \frac{V}{R \vee 1} \mathbbm{1}_{\widehat{k}_{+} > 0} \right] \\
&= \mathbb{E} \left[
\frac{
\#\left\{ j \leq \widehat{k}_{+} : p_j \leq \frac{1}{2},\ j \in \mathcal{H} \right\}
}{
1 + \#\left\{ j \leq \widehat{k}_{+} : p_j > \frac{1}{2},\ j \in \mathcal{H} \right\}
}
\right. \\
&\quad \left. \times
\left(
\frac{
1 + \#\left\{ j \leq \widehat{k}_{+} : p_j > \frac{1}{2},\ j \in \mathcal{H} \right\}
}{
\#\left\{ j \leq \widehat{k}_{+} : p_j \leq \frac{1}{2} \right\} \vee 1
}
\right)
\mathbbm{1}_{\widehat{k}_{+} > 0}
\right] \\
&\leq \mathbb{E} \left[
\frac{
\#\left\{ j \leq \widehat{k}_{+} : p_j \leq \frac{1}{2},\ j \in \mathcal{H} \right\}
}{
1 + \#\left\{ j \leq \widehat{k}_{+} : p_j > \frac{1}{2},\ j \in \mathcal{H} \right\}
}
\right] q \leq q.
\end{align*}

where the first inequality holds by the definition of $\widehat{k}_{+}$ and the second inequality holds by Lemma \ref{Lamma:Lemb2}.

Similarly, for the knockoff threshold, we have $V=\#\{j\leq \widehat{k}_0: p_j\leq 1/2, j\in \mathcal{H}\}$ and $R=\#\{j\leq \widehat{k}_0:p_j\leq 1/2\}$ where $\widehat{k}_0$ satisfies that $|W_{\widehat{k}_0}|=\tau$ where $\tau$ is defined as in theorem 1, then

\begin{align*}
\mathbb{E} \left[ \frac{V}{R + q^{-1}} \right]
&= \mathbb{E} \left[
\frac{
\#\left\{ j \leq \widehat{k}_0 : p_j \leq \frac{1}{2},\ j \in \mathcal{H} \right\}
}{
1 + \#\left\{ j \leq \widehat{k}_0 : p_j > \frac{1}{2},\ j \in \mathcal{H} \right\}
}
\right. \\
&\quad \left. \times 
\left(
\frac{
1 + \#\left\{ j \leq \widehat{k}_0 : p_j > \frac{1}{2},\ j \in \mathcal{H} \right\}
}{
\#\left\{ j \leq \widehat{k}_0 : p_j \leq \frac{1}{2} \right\} + q^{-1}
}
\right)
\mathbbm{1}_{\widehat{k}_0 > 0}
\right] \\
&\leq \mathbb{E} \left[
\frac{
\#\left\{ j \leq \widehat{k}_0 : p_j \leq \frac{1}{2},\ j \in \mathcal{H} \right\}
}{
1 + \#\left\{ j \leq \widehat{k}_0 : p_j > \frac{1}{2},\ j \in \mathcal{H} \right\}
}
\right] q \leq q.
\end{align*}

where the first inequality holds by the definition of $\widehat{k}_0$ and the second inequality holds by Lemma \ref{Lamma:Lemb2}.
\end{proof}

\subsection*{B.2: Proof of Theorem 1}\label{pf:thm1}
\begin{proof}
To prove Theorem 1, we first show that by the construction of the Model-X knockoff and the choice of compatible test statistics, we have $[Z^a_j, \tilde{Z}^a_j] \eqd [\tilde{Z}^a_j, Z^a_j]$ for all $j\in \mathcal{H}_a$ and $[Z^b_j, \tilde{Z}^b_j] \eqd [\tilde{Z}^b_j, Z^b_j]$ for all $j\in \mathcal{H}_b$. Since we can include $\V$ in the knockoff construction to make $[M_j, \tilde{M}_j,\V]\eqd [\tilde{M}_j, M_j,\V]$ for any $j\in \mathcal{H}_a$, we have the statistics from the marginal regression of $X$ on $[M_j, \tilde{M}_j, \V]$ have the form of $$[Z^a_j, \tilde{Z}^a_j]=g([M_j, \tilde{M}_j, \V]^\top [M_j, \tilde{M}_j, \V], [M_j, \tilde{M}_j, \V]^\top X).$$ Since $[M_j, \tilde{M}_j, \V]\eqd [\tilde{M}_j, M_j, \V]$. For $j\in \mathcal{H}_a$, we have $[M_j, \tilde{M}_j]X=0$, so we have $Z^a_j\eqd \tilde{Z}^a_j$.

Also, since we have $[\M, \tilde{\M}, \V]\eqd [\M, \tilde{\M},\V]_{Swap(S)}$ for any $S\subset \mathcal{H}_b$, the Model-X knockoff results from \cite{candes2018} and \cite{dai2024_mediation} can be applied to obtain the exchangeability for $[Z^b_j, \tilde{Z}^b_j]$. When linear model is assumed, we can make a projection of $Y$, $X$, and $M_j$ on the orthogonal space of $\V$ for path-b. So the knockoff generated using residuals based on fixed knockoff construction and test statistics that compatible with fixed knockoff construction satisfy the exchangeability requirement for path-b by \cite{barber2015}.

Since we use data splitting, we have $[\Z^a,\tilde{\Z}^a]$ and $[\Z^b,\tilde{\Z}^b]$ are independent. Define p-values $p_j=1$ if $W_j<0$ and $p_j=1/2$ if $W_j>0$, then Lemma \ref{Lamma:Lemb1} guarantee that we have $p_j\geq \text{Uniform}[0,1]$ and are independent from nonnulls. Since we assumed that the path-a statistics satisfy the weak dependence assumption, we have $p_j$s are weakly correlated. Applying Lemma \ref{Lamma:Lemb3}, we obtain the mFDR and FDR control results. 
\end{proof}

\section*{Web Appendix C: Additional Simulations}
\subsection*{C.1: Details of Data Generation and Simulation Settings}
\paragraph{\textbf{Linear Models Settings}}\label{sec:linear}

{\color{black} We first performed simulations under linear model settings and sampled the data with sample size $n=1000$. We designate $M_{i,10}, M_{i,11}, M_{i,12}, \cdots, M_{i,30}$ as the true mediators, meaning they have both non-zero $\alpha$ and $\beta$ coefficients. To achieve this, in path-a, we assign a non-zero value $a$ to the first 9 entries of $\bm{\alpha}$, $0.7a$ to the entries 10:30 of $\bm{\alpha}$. In path-b, we assign a non-zero value $b$ to the 31st through 40th coefficients of $\bm{\beta}$, $0.7b$ to the 10th through 30th coefficients of $\bm{\beta}$, and we set all other coefficient values to 0. 
Our parameters are $\bm{\alpha}, \bm{\beta}$ defined as follows:\\
$
\bm{\alpha} = (\, \underbrace{a, \ldots, a}_{9},\ \underbrace{0.7a, \ldots, 0.7a}_{21},\ \underbrace{0, \ldots, 0}_{p - 30} \,)^\top
$,\\
$
\bm{\beta} = (\, \underbrace{0, \ldots, 0}_{9},\ \underbrace{0.7b, \ldots, 0.7b}_{21},\ \underbrace{b, \ldots, b}_{10},\ \underbrace{0, \ldots, 0}_{p - 40} \,)^\top$.

In the simulation setting, we vary the correlation structure among mediators by setting the correlation coefficient $\rho \in \{0,0.1,\dots,0.9\}$. To examine the impact of signal strength in path a and path b, we vary $a \in \{0.1,0.2,0.3,0.4,0.5\}$ and $b \in \{0.1,0.2,0.3,0.4,0.5\}$.
The number of potential mediators is set to vary across $p=100,200,\dots,500$ to evaluate performance in different dimensional settings. We fix $\gamma=1$,$\bm{\eta}_{1}=(1,\cdots,1)^\top$, $\eta_{2}=1$ and $\sigma=0.4$.

\paragraph{\textbf{Non-linear Models Settings}} \label{sec:nonlinear}
We also performed simulations under non-linear model settings. We consider two settings: (1) Two-way interactions exist in path-b; (2) The outcome $\bm{Y}$ is modeled as the element-wise cosine transformation of the mediator matrix $\M$ in path-b. We sampled the data with sample size $n=1000$.

\paragraph{\textbf{Setting 1}: Two-way interactions} 


We also utilize the same true mediators defined in the linear setting, with $M_{i,10}, \cdots, M_{i,30}$ as the true mediators. In path-a, we assign a non-zero value $0.5a$ to the first 9 entries of $\bm{\alpha}$, $a$ to the 10:30th entries of $\bm{\alpha}$. 
In path-b, we assign a non-zero value $\boldsymbol{\delta}= b \times \bf{1}_{15} $
to ensure that interactions among the 10th through the 39th mediators represent true signals for path-b. Specifically, we define $\mathbf{Z}$, the matrix of interaction terms, as
\[
\mathbf{Z} = \begin{bmatrix}
M_{1,10}M_{1,11} & M_{1,12}M_{1,13} & \cdots & M_{1,38}M_{1,39} \\
M_{2,10}M_{2,11} & M_{2,12}M_{2,13} & \cdots & M_{2,38}M_{2,39} \\
\vdots & \vdots & \ddots & \vdots \\
M_{n,10}M_{n,11} & M_{n,12}M_{n,13} & \cdots & M_{n,38}M_{n,39}
\end{bmatrix},
\]
where each row corresponds to interaction terms for a given subject i, and we assign non-zero coefficients $\mathbf{\delta}$ to these interaction terms. We define $
\bm{\alpha} = (\underbrace{0.5a, \ldots, 0.5a}_{9},\ \underbrace{a, \ldots, a}_{21},\ \underbrace{0, \ldots, 0}_{p - 30})
$,
$
\boldsymbol{\delta} = (\underbrace{0, \ldots, 0}_{4},\ \underbrace{b, \ldots, b}_{15},\ \underbrace{0, \ldots, 0}_{p/2 - 19})
$

In the simulation setting, we vary the correlation structure among mediators by setting the correlation coefficient $\rho \in \{0.1, 0.3, 0.5\}$. To examine the impact of signal strength in path a and path b, we vary $a \in \{0.4, 0.6\}$ and $b \in \{0.5, 1\}$.
The number of potential mediators is set to be $p=100$. We fix $\gamma=1$,$\bm{\eta}_{1}=(0,\cdots,0)^\top$, and $\sigma=0.4$.


\paragraph{\textbf{Setting 2}: Cosine transformation}
We extend our model to a nonlinear model by applying the cosine function to each element of $\M$. We also set the sample size to $n=1000$ with a higher dimension of mediators $p=400$. 

We still set {$M_{i,10}, \cdots, M_{i,30}$} as the true mediators. 
In path-a, we assign a non-zero value $a$ to the first 9 coefficients of $\alpha$, $0.7a$ to the coefficient 10:30 of $\bm{\alpha}$. In path-b, we assign a non-zero value of $b$ to the 10th through 30th coefficients of $\bm{\beta}$, $0.7b$ to the 31st through 39th coefficients of $\bm{\beta}$, and we set all other coefficient values to 0. 

Our parameters are defined as follows:\\
$
\bm{\alpha} = (\underbrace{a, \ldots, a}_{9},\ \underbrace{0.7a, \ldots, 0.7a}_{21},\ \underbrace{0, \ldots, 0}_{p - 30}).
$,\\
$
\bm{\beta} = (\underbrace{0, \ldots, 0}_{9},\ \underbrace{b, \ldots, b}_{21},\ \underbrace{0.7b, \ldots, 0.7b}_{9},\ \underbrace{0, \ldots, 0}_{p - 39}).
$
In the simulation setting, we vary the correlation structure among mediators by setting the correlation coefficient $\rho \in \{0.1, 0.3\}$. To examine the impact of signal strength in path a and path b, we vary $a \in \{0.2, 0.5\}$ and $b \in \{0.5, 1\}$.
The number of potential mediators is set to be $p=400$. We fix $\gamma=1$,$\bm{\eta}_{1}=(0,\cdots,0)^\top$, and $\sigma=0.4$.

\subsection*{C.2: Details of Method specification}
For our proposed method, for path-a, we use the Guassian Model-X method to generate knockoffs $(\tilde{M}_j, \tilde{V})$ for $(M_j, V)$, and fit a regression model of the form $X \sim (M_j, \tilde{M}_j, V)$ to obtain test statistics $(\mathbf{Z}^a, \tilde{\mathbf{Z}}^a)$ using the absolute value of the coefficients (more details are in Apendix A). For path-b linear settings, we use the PLS option in Remark 1 to derive the test statistics $(\mathbf{Z}^b, \tilde{\mathbf{Z}}^b)$. For the nonlinear setting, we fit random forests models and extract variable importance scores to construct $(\mathbf{Z}^b, \tilde{\mathbf{Z}}^b)$. For HIMA methods, (i.e.,HIMA, HIMA2 and eHIMA), we use the HIMA package version 2.2.2. with their default methods. We set the controlled FDR level $q=0.2$ for our proposed method and all HIMA methods. For method MediationFDR, we apply the MediationFDR package version 1.0.0 with the default methods and to control the overall FDR at $0.2$, we set $q_1=0.1$ in step 1 and $q_2=0.1$ in step 4. For the method MCP, we use the package MultiMed version 2.24.0 with the default methods and set the thresholds as 0.1 for both path-a and path-b.  

}
\subsection*{C.3: Additional Simulation Results}
We extend our non-linear settings to evaluate the performance. For \textbf{Setting 1}, we explore the different non-zero path-a coefficient values: $a = 0.4,0.6,0.8,1.0$; different non-zero path-b values,$b = 0.5,1,1.5,2$; and different correlation coefficients of $\bm{M}$: $\rho = 0,0.1,\dots,0.9$. For \textbf{Setting 2}, we explore the different non-zero path-$a$ coefficient values: $a = 0.1,0.2,0.3,0.4,0.5$; different non-zero path-$b$ coefficient value: $b = 0.5,0.7,\dots,1.5$; different correlation coefficients of $\M$: $\rho = 0,0.1,\dots,0.7$.

Figure \ref{fig:chg_nonlinear} summarizes the empirical results for non-linear models with interactions (\textbf{Setting 1}) and models where mediators undergo a cosine transformation. (\textbf{Setting 2}).  In Setting 1, we observe that GKMS consistently achieves the highest power across all examined conditions, particularly under low to moderate correlation levels (\(\rho \leq 0.5\)). As the correlation among mediators increases, the power of all methods declines; however, GKMS demonstrates a more gradual decrease and remains superior throughout. When varying the signal strengths \(a\) (path-a) and \(b\) (path-b), GKMS continues to dominate in terms of power, while maintaining FDR well below or near the nominal level of 0.1. In contrast, eHIMA often fails to control FDR, particularly under weak signals or high correlation, and shows substantially reduced power.

Setting 2 presents a more challenging scenario with higher dimensionality and weaker path-a signal (\(a=0.2\)). Even in this difficult setting, GKMS remains the most powerful method, significantly outperforming existing approaches. Its power is robust to increases in correlation and varying signal strength in both paths. FDR remains controlled across all conditions. Competing methods, particularly eHIMA, display inflated FDR and poor power in most cases, while MediationFDR and HIMA variants demonstrate more conservative behavior with limited signal detection capabilities.

Across both settings, GKMS achieves strong FDR control and substantial improvements in power. This confirms its effectiveness in nonlinear and high-dimensional environments. The method is robust across a wide range of conditions, including increasing correlation and diminishing effect sizes, making it a promising tool for complex mediation analysis involving nonlinear or interaction effects.



\section*{Web Appendix D: Additional ADNI data information}
\subsection*{Appendix D.1: Structural MRI-Derived Measures}

The 555 candidate structural MRI-derived measures include features from cortical surface regions (aparc.a2009s.stats) in the left (lh) and right (rh) hemispheres, white matter (wmparc.stats), and subcortical structures (aseg.stats). For cortical surface measurements, we focus on three specific metrics: Thickness Average (ThickAvg), Mean Curvature (MeanCurv), and Gaussian Curvature (GausCurv). For white matter and subcortical structures, we focus on volumetric measures. These measures are further categorized by different anatomical regions within the cortical surface, white matter, and subcortical structures. The resulting combined names for the 555 MRI-derived measures incorporate the structure type, measurement type, and specific region.
The combined names are listed in the github: \url{https://github.com/RunqiuWang22/Generalized_Knockoff_Mediator_Selection}.

\subsection*{Appendix D.2: Cognitive Outcome Measures}

The multiple scores evaluate different aspects of cognitive function, including the Clinical Dementia Rating Sum of Boxes (CDR-SB), Alzheimer’s Disease Assessment Scale-Cognitive Subscale (ADAS-Cog) and Mini-Mental State Examination (MMSE). These scores are obtained through direct assessments performed by clinicians or neuropsychologists, providing comprehensive evaluations of brain functions such as memory, orientation, problem-solving, and community engagement. CDR-SB evaluates six domains of cognitive and functional performance, with higher scores indicating greater impairment \citep{lynch2005clinical}. ADAS-Cog, includes versions such as ADAS-11, ADAS-13, and ADAS-Q4, and assesses memory, language, and praxis, with the higher scores indicating more severe cognitive dysfunction \citep{podhorna2016alzheimer}. Additionally, the MMSE is a global screening tool measuring orientation, attention, memory, and language, with lower scores reflecting greater cognitive impairment \citep{pike2017sex}.

\subsection*{Appendix D.3: Estimation of Mediation Effects}
{\color{black}For coefficient estimation, we employed structural equation modeling (SEM) through the \texttt{lavaan} R package \citep{rosseel_lavaan_2012} for methods assuming linear relationships (all except GKMS(RF)). For each method, we use the identified mediators along with the exposure variable ($X$), outcome ($Y$), and covariate (age), and construct a path model where in each mediator $M_k$ is regressed on $X$ and $V$, while the outcome $Y$ was regressed on $X$, $V$, and all selected mediators.

This SEM framework facilitated estimation of:
\begin{itemize}
    \item The \textbf{path-a coefficients} ($\alpha_k$): representing effects of $X$ on each mediator $M_k$,
    \item The \textbf{path-b coefficients} ($\beta_k$): capturing effects of each mediator $M_k$ on $Y$ (adjusting for $X$ and $V$),
    \item The \textbf{individual mediation effects} ($\text{indirect}_k = \alpha_k \times \beta_k$),
\end{itemize}

All coefficients were estimated with corresponding standard errors, 95\% confidence intervals, and $p$-values derived via asymptotic normal theory.

For GKMS(RF), we implemented a hybrid approach to causal mediation analysis accommodating linear path-a and nonlinear path-\textit{b} relationships. Path-a coefficients are estimated the same as above. Path-b coefficients ($\beta$) were quantified using Random Forest variable importance measures with permutation-based significance testing, allowing for detection of complex nonlinear relationships without requiring pre-specified functional forms.

Effect directionality was determined through partial dependence analysis comparing outcome predictions at the first and third quartiles of mediator values, with bootstrap confidence intervals (100 resamples). For mediator-specific natural indirect effects ($\widehat{\text{NIE}}_k$), we use g-formula to compute it. To evaluate the integration involved in the g-formula, we implemented the three-step Monte Carlo integration approach described by \citet{imai2010general}, which appropriately addresses the counterfactual framework through:
\begin{enumerate}
    \item Sampling mediator values under alternative treatment conditions;
    \item Predicting counterfactual outcomes across all treatment-mediator value combinations;
    \item Computing each mediator-specific effect $\widehat{\text{NIE}}_k$ as 
    \[
    \mathbb{E}[Y(0, M_k(1), M_{-k}(0))] - \mathbb{E}[Y(0, M_k(0), M_{-k}(0))].
    \]
\end{enumerate}
Bootstrap resampling (100 replicates) provided confidence intervals and p-values for these mediation effects. 
}

{\color{black}
\subsection*{Appendix D.4: Additional Simulation Results}
As a sensitivity analysis, we also applied our proposed method with Lasso linear regression in path-b (GKMS (Lasso)), HIMA, HIMA2, eHIMA, and MediationFDR methods for comparison. Results summarize in Tables \ref{tab:CDRSB}, \ref{tab:ADAS}, and \ref{tab:MMSE} for CDR-SB, ADAS-Cog, and MMSE scores, respectively. The mediation effects were estimated using structural equation modeling (SEM) through the \texttt{lavaan} R package \citep{rosseel_lavaan_2012}. 

The results of using Lasso method in path-b showed partial overlap with those selected by the random forest. Lasso consistently selected hippocampal and amygdala volumes as mediators across all three cognitive outcomes. For the CDR-SB, Lasso identified left hippocampal volume, bilateral amygdala volumes, right inferior lateral ventricle volume, and brainstem volume as significant mediators.
For the ADAS-Cog, selected mediators by Lasso included left hippocampal and amygdala volumes, left angular gyrus cortical thickness, right superior temporal sulcus cortical thickness, and left choroid plexus volume. Finally, for the MMSE, mediators identified by Lasso comprised left hippocampal volume, bilateral amygdala volumes, right inferior lateral ventricle volume, and left medial orbitofrontal cortex volume.

Comparatively, HIMA, eHIMA, and MediationFDR selected fewer mediators but consistently included the amygdala volumes. Moreover, HIMA2 did not select any mediators for CDR-SB or MMSE but identified two measures for ADAS-Cog. In general, almost all the methods identified core mediators, including hippocampal and amygdala structures, variations occur in additional cortical and subcortical regions, reflecting methodological differences in mediator selection.
}
\begin{table}[!p]
\caption{ADNI study mediator selection results for the \textbf{Clinical Dementia Rating Sum of Boxes (CDR-SB)}. (L): Left; (R): Right.} 
\label{tab:CDRSB} 
\centering 
\begin{threeparttable}
\begingroup
  \renewcommand{\arraystretch}{0.9}      
  \setlength{\tabcolsep}{5pt}            
  \fontsize{8}{9.5}\selectfont           
\begin{tabular}{l l l} 
\hline 
Method & Selected Mediator & Effect (95\% CI; p-value)\textsuperscript{†} \\ [0.5ex] 
\hline 
GKMS (Lasso)
 & Volume of Hippocampus (L) & $\alpha$ = -0.52 (-0.627, -0.414; $<$0.001) \\ 
 & & $\beta$ = -0.132 (-0.193; -0.071;$<$0.001) \\ 
 & & Mediation = 0.069 (0.034, 0.104; $<$0.001) \\
 & & \\
&  Volume of Amygdala (L) & $\alpha$ = -0.505 (-0.614, -0.397; $<$0.001) \\ 
 & & $\beta$ = -0.178 (-0.238; -0.118;$<$0.001) \\ 
 & & Mediation = 0.09 (0.054, 0.126; $<$0.001) \\
 & & \\
&  Volume of Amygdala (R) & $\alpha$ = -0.596 (-0.706, -0.487; $<$0.001) \\ 
 & & $\beta$ = -0.111 (-0.171; -0.052;$<$0.001) \\ 
 & & Mediation = 0.066 (0.029, 0.104; 0.001) \\
 & & \\
&  Volume of Inf.Lat.Vent\textsuperscript{1} (R) & $\alpha$ = -0.483(-0.591, -0.376; $<$0.001) \\ 
 & & $\beta$ = 0.15 (0.089; 0.211;$<$0.001) \\ 
 & & Mediation = -0.072 (-0.106, -0.039; $<$0.001) \\
 & & \\
 &   Volume of Brain Stem & $\alpha$ = -0.955(-1.061, -0.848; $<$0.001) \\ 
 & & $\beta$ = -0.121 (0.06; 0.183;$<$0.001) \\ 
 & & Mediation = -0.116 (-0.176, -0.056; $<$0.001) \\
 & & \\
MediationFDR
&  Volume of Amygdala (L) & $\alpha$ = -0.505 (-0.614, -0.397; $<$0.001) \\ 
 & & $\beta$ = -0.376 (-0.438; -0.314;$<$0.001) \\ 
 & & Mediation = 0.19 (0.139, 0.241; $<$0.001) \\
& &\\
HIMA
&  Volume of Amygdala (L) & $\alpha$ = -0.505 (-0.614, -0.397; $<$0.001) \\ 
 & & $\beta$ = -0.282 (-0.343, -0.222; $<$0.001) \\ 
 & & Mediation = 0.143 (0.099, 0.186; $<$0.001) \\
 & & \\
&  Volume of Inf.Lat.Vent (R) & $\alpha$ = -0.483(-0.591, -0.376; $<$0.001) \\ 
 & & $\beta$ = 0.135 (0.074, 0.196; $<$0.001)\\
 & & Mediation = -0.065 (-0.098, -0.032; $<$0.001) \\
 & & \\
& ThickAvg\textsuperscript{2} of Parahip\textsuperscript{3} (R) & $\alpha$ = 0.123 (0.007, 0.24; 0.037) \\ 
& & $\beta$ = -0.128 (-0.185, -0.072; $<$0.001)\\
& & Mediation = -0.016 (-0.032, 0.001; 0.059) \\
& & \\
eHIMA
&  Volume of Amygdala (L) & $\alpha$ = -0.505 (-0.614, -0.397; $<$0.001) \\ 
 & & $\beta$ = -0.376 (-0.438; -0.314;$<$0.001) \\ 
 & & Mediation = 0.19 (0.139, 0.241; $<$0.001) \\
\hline 
\end{tabular}

\begin{tablenotes}
         \item {\color{black}\textsuperscript{†} Coefficient estimation is performed using structural equation modeling (SEM) via the \texttt{lavaan} R package \citep{rosseel_lavaan_2012}.}  
         \item Abbreviations (all extracted from the aparc.a2009s statistics file):
        \item \textsuperscript{1} Inf.Lat.Vent: Inferior Lateral Ventricle.
        \item \textsuperscript{2} ThickAvg: Average cortical thickness of the specified cortical region.
        \item \textsuperscript{3} Parahip: Parahippocampal gyrus, a medial temporal lobe region.
     \end{tablenotes}
\endgroup
  \end{threeparttable}
\end{table}

\begin{table}[ht]
\small
\caption{ADNI study mediator selection results for the\textbf{Alzheimer’s Disease Assessment Scale-Cognitive Subscale (ADAS-Cog)}. (L): Left; (R): Right.} 
\label{tab:ADAS} 
\centering 
\begin{threeparttable}
\begingroup
  \renewcommand{\arraystretch}{0.9}      
  \setlength{\tabcolsep}{5pt}            
  \fontsize{8}{9.5}\selectfont           
\begin{tabular}{l l l} 
\hline 
Method & Selected Mediator & Effect (95\% CI; p-value)\textsuperscript{†}  \\ [0.5ex] 
\hline 
GKMS (Lasso)
 & Volume of Hippocampus (L) &  $\alpha$ = -0.52 (-0.627, -0.414; $<$0.001) \\ 
 & & $\beta$ = -0.083 (-0.145, -0.02; 0.009) \\ 
 & & Mediation = 0.043 (0.009, 0.077; 0.012) \\
 & &\\
&  Volume of Amygdala (L) & $\alpha$ = -0.505 (-0.614, -0.397; $<$0.001) \\ 
 & & $\beta$ = -0.139 (-0.201, -0.078; $<$0.001) \\ 
 & & Mediation = 0.07 (0.036, 0.105; $<$0.001) \\
 & &\\
&  ThickAvg\textsuperscript{1} of Angular Gyrus(L) & $\alpha$ = 0.233 (0.118, 0.349; $<$0.001) \\ 
 & & $\beta$ = -0.099 (-0.157, -0.041; 0.001) \\ 
 & & Mediation = -0.023 (-0.041, -0.005; 0.011) \\
 & & \\
 &   ThickAvg of S.temporal.sup\textsuperscript{2}(R)  & $\alpha$ = 0.092 (-0.023, 0.207; 0.117) \\ 
 & & $\beta$ = -0.167 (-0.226, -0.109; $<$0.001) \\ 
 & & Mediation = -0.015 (-0.035, 0.005; 0.131) \\
 & &\\
 &   Volume of Choroid Plexus (L) & $\alpha$ = -0.719 (-0.826, -0.611; $<$0.001) \\ 
 & & $\beta$ = 0.081 (0.019, 0.143; 0.011) \\ 
 & & Mediation = -0.058 (-0.104, -0.013; 0.012) \\
 & &\\
HIMA
&  Volume of Amygdala (L) & $\alpha$ = -0.505 (-0.614, -0.397; $<$0.001) \\ 
 & & $\beta$ = -0.324 (-0.387, -0.261; $<$0.001) \\ 
 & & Mediation = 0.164 (0.116, 0.211; $<$0.001) \\
& &\\
eHIMA
&   Volume of Amygdala (L) & $\alpha$ = -0.505 (-0.614, -0.397; $<$0.001) \\ 
& & $\beta$ = -0.324 (-0.387, -0.261; $<$0.001) \\ 
& & Mediation = 0.164 (0.116, 0.211; $<$0.001) \\
& &\\
HIMA2
& ThickAvg of Cuneus Gyrus(L) & $\alpha$ = -0.202 (-0.322, -0.081; 0.001) \\ 
 & & $\beta$ = -0.063 (-0.123, -0.004; 0.036) \\ 
 & & Mediation = 0.013 (-0.001, 0.027; 0.078) \\
 & &\\
&  ThickAvg of Pole temporal(R) & $\alpha$ = 0.259 (0.143, 0.376; $<$0.001) \\ 
 & & $\beta$ = -0.084 (-0.146, -0.023; 0.007) \\ 
 & & Mediation = -0.022 (-0.041, -0.003; 0.022) \\ 
\hline 
\end{tabular}
\begin{tablenotes}
       \item [] We use version ADAS-13 representing the results of ADAS-Cog.
        \item {\color{black}\textsuperscript{†} Coefficient estimation is performed using structural equation modeling (SEM) via the \texttt{lavaan} R package \citep{rosseel_lavaan_2012}.}     
        \item Abbreviations (all extracted from the aparc.a2009s statistics file):
        \item \textsuperscript{1} ThickAvg: Average cortical thickness of the specified cortical region.
        \item \textsuperscript{2} S.temporal.sup: Superior temporal sulcus.
     \end{tablenotes}
\endgroup
  \end{threeparttable}
\end{table}

\begin{table}[!p]
\small
\caption{ADNI study mediator selection results for the \textbf{Mini-Mental State Examination (MMSE)}. (L): Left; (R): Right.} 
\label{tab:MMSE} 
\centering 
\begin{threeparttable}
\begingroup
  \renewcommand{\arraystretch}{0.9}      
  \setlength{\tabcolsep}{5pt}            
  \fontsize{8}{9.5}\selectfont    
\begin{tabular}{l l l} 
\hline 
Method & Selected Mediator & Effect (95\% CI; p-value)\textsuperscript{†}  \\ [0.5ex] 
\hline 
GKMS (Lasso)
 & Volume of Hippocampus (L) & $\alpha$ = -0.52 (-0.627, -0.414; $<$0.001) \\ 
 & & $\beta$ = 0.111 (0.048, 0.174; 0.001) \\ 
 & & Mediation = -0.058 (-0.093, -0.023; 0.001) \\
 & & \\
&  Volume of Amygdala (L) & $\alpha$ = -0.505 (-0.614, -0.397; $<$0.001) \\ 
 & & $\beta$ = 0.184 (0.122, 0.246; $<$0.001) \\ 
 & & Mediation = -0.093 (-0.13, -0.056; $<$0.001) \\
 & & \\
&  Volume of Amygdala (R) & $\alpha$ = -0.596 (-0.706, -0.487; $<$0.001) \\ 
 & & $\beta$ = 0.068 (0.007, 0.129; 0.028) \\ 
 & & Mediation = -0.041 (-0.078, -0.004; 0.032) \\
 & &\\
&  Volume of Inf.Lat.Vent\textsuperscript{1} (R) & $\alpha$ = -0.483(-0.591, -0.376; $<$0.001) \\ 
 & & $\beta$ = -0.088 (-0.15, -0.026; 0.006) \\ 
 & & Mediation = 0.042 (0.011, 0.074; 0.008) \\
 & & \\
 &   Volume of Medial Orbitofrontal (L) & $\alpha$ = -0.702 (-0.816, -0.588; $<$0.001) \\ 
 & & $\beta$ = -0.147 (-0.206, -0.088; $<$0.001) \\ 
 & & Mediation = 0.103 (0.059, 0.148; $<$0.001) \\
 & & \\
MediationFDR
&  Volume of Amygdala (L) & $\alpha$ = -0.505 (-0.614, -0.397; $<$0.001) \\ 
 & & $\beta$ = 0.336 (0.273, 0.399; $<$0.001)\\ 
 & & Mediation = -0.17 (-0.218, -0.121; $<$0.001) \\
& &\\
HIMA
&   Volume of Amygdala (L) & $\alpha$ = -0.505 (-0.614, -0.397; $<$0.001) \\ 
 & & $\beta$ = 0.219 (0.158, 0.281; $<$0.001) \\ 
 & & Mediation = -0.111 (-0.15, -0.072; $<$0.001) \\
 & &\\
&  ThickAvg\textsuperscript{2} of S.temporal.sup\textsuperscript{3} (L) & $\alpha$ = 0.12 (0.007, 0.233; 0.037) \\ 
 & & $\beta$ = 0.246 (0.187, 0.305; 0)\\
 & & Mediation = 0.03 (0.001, 0.058; 0.043) \\
& &\\
eHIMA
&   Volume of Amygdala (L) & $\alpha$ = -0.505 (-0.614, -0.397; $<$0.001) \\ 
 & & $\beta$ = 0.336 (0.273, 0.399; $<$0.001)\\ 
 & & Mediation = -0.17 (-0.218, -0.121; $<$0.001) \\
\hline 
\end{tabular}
\begin{tablenotes}
        \item {\color{black}\textsuperscript{†} Coefficient estimation is performed using structural equation modeling (SEM) via the \texttt{lavaan} R package \citep{rosseel_lavaan_2012}.}  
        \item Abbreviations (all extracted from the aparc.a2009s statistics file):
        \item \textsuperscript{1} Inf.Lat.Vent: Inferior Lateral Ventricle.
        \item \textsuperscript{2} ThickAvg: Average cortical thickness of the specified cortical region.
        \item \textsuperscript{3} S.temporal.sup: Sperior temporal sulcus.       
     \end{tablenotes}
\endgroup
  \end{threeparttable}
\end{table}

\begin{figure}
    \centering
     \includegraphics[scale=0.7]{Figure/legend_only_0.20_0.50_0.10_0.70_0.40_S2_CS.pdf}\\
    \includegraphics[scale=0.33]{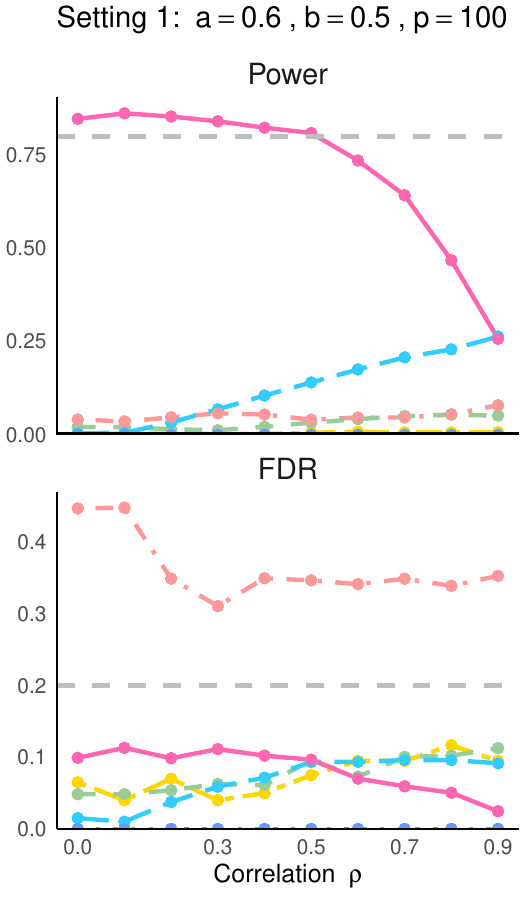}
    \includegraphics[scale=0.33]{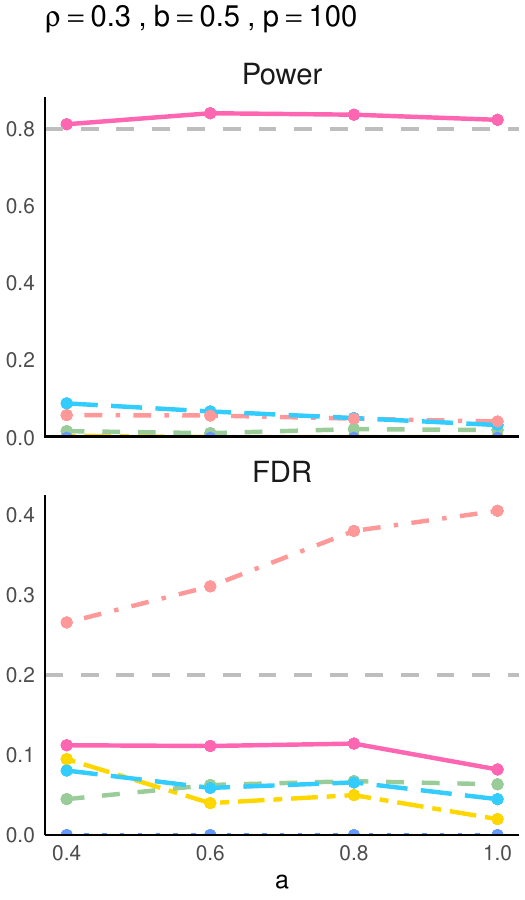}
      \includegraphics[scale=0.33]{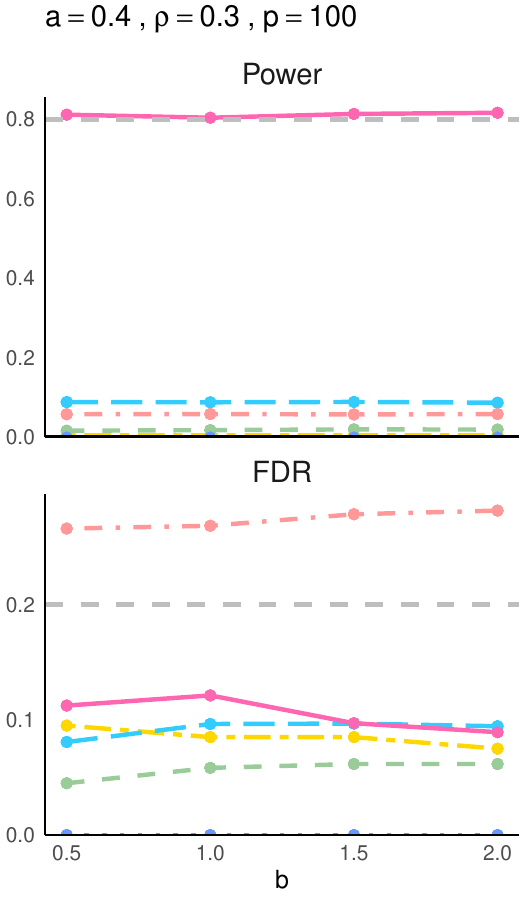}\\
    \includegraphics[scale=0.33]{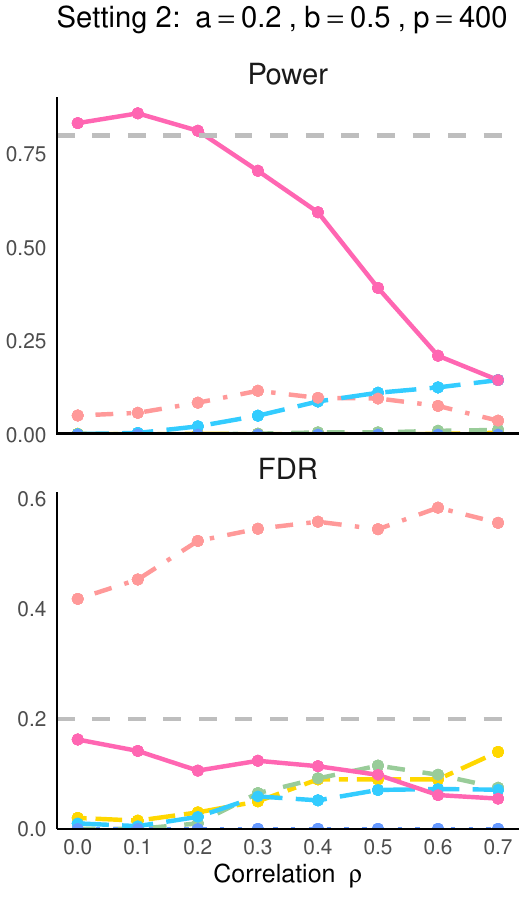}
    \includegraphics[scale=0.33]{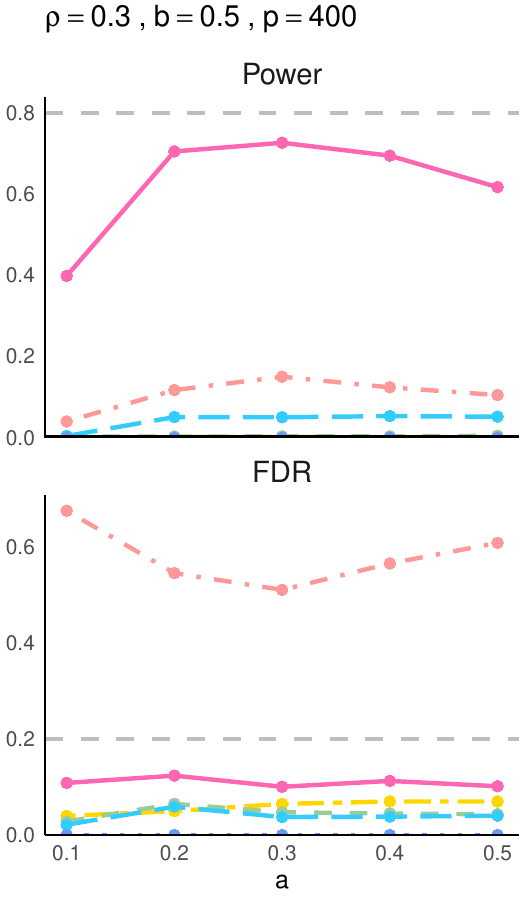}
      \includegraphics[scale=0.33]{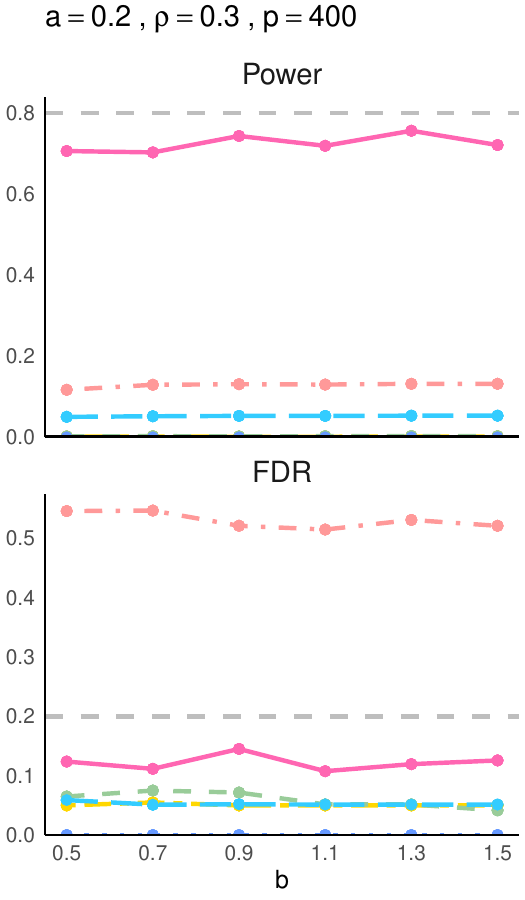}\\
   \caption{Simulation results for the empirical power and false discovery rate (FDR) with sample size $n=1000$ are shown for correlated mediators with a defined correlation structure across 100 replications. The results are stratified by the effect size of the correlation structure ($\rho$),
   the number of potential mediators (p), the effect size of path-a (a), and the effect size of path-b (b) in the non-linear models with Setting 1 in the upper, Setting in the lower. }
   \label{fig:chg_nonlinear}
\end{figure}

\end{document}